\documentclass[english]{sn-jnl}
\usepackage[T1]{fontenc}
\usepackage[latin9]{inputenc}
\usepackage{array}
\usepackage{multirow}
\usepackage{amsmath}
\usepackage{amsthm}
\usepackage{amssymb}
\usepackage{stackrel}
\usepackage{graphicx}

\makeatletter

\newcommand*\LyXThinSpace{\,\hspace{0pt}}
\providecommand{\tabularnewline}{\\}

\theoremstyle{plain}
\newtheorem{assumption}{\protect\assumptionname}
\theoremstyle{definition}
\newtheorem{defn}{\protect\definitionname}
\theoremstyle{plain}
\newtheorem{lem}{\protect\lemmaname}
\theoremstyle{plain}
\newtheorem{thm}{\protect\theoremname}

\makeatother

\usepackage{babel}
\providecommand{\assumptionname}{Assumption}
\providecommand{\definitionname}{Definition}
\providecommand{\lemmaname}{Lemma}
\providecommand{\theoremname}{Theorem}

\begin{document}
\title{Collaborative Spacecraft Servicing under Partial Feedback using Lyapunov-based
Deep Neural Networks}

\author*[1]{\fnm{Cristian F.} \sur{Nino}}\email{cristian1928@ufl.edu}
\author[1]{\fnm{Omkar Sudhir} \sur{Patil}}\email{patilomkarsudhir@ufl.edu}
\equalcont{These authors contributed equally to this work.}
\author[1]{\fnm{Christopher D.} \sur{Petersen}}\email{c.petersen1@ufl.edu}
\equalcont{These authors contributed equally to this work.}
\author[2]{\fnm{Sean} \sur{Phillips}}\email{sean.phillips.9@spaceforce.mil}
\equalcont{These authors contributed equally to this work.}
\author[1]{\fnm{Warren E.} \sur{Dixon}}\email{wdixon@ufl.edu}
\equalcont{These authors contributed equally to this work.}

\affil*[1]{\orgdiv{Department of Mechanical and Aerospace Engineering}, \orgname{University of Florida}, \orgaddress{\city{Gainesville}, \state{FL}, \postcode{32611}, \country{USA}}}
\affil[2]{\orgdiv{Air Force Research Laboratory, Space Vehicles Directorate}, \orgaddress{\city{Kirtland AFB}, \state{NM}, \postcode{87117}, \country{USA}}}

\abstract{Multi-agent systems are increasingly applied in space missions, including distributed space systems, resilient constellations, and autonomous rendezvous and docking operations. A critical emerging application is collaborative spacecraft servicing, which encompasses on-orbit maintenance, space debris removal, and swarm-based satellite repositioning. These missions involve servicing spacecraft interacting with malfunctioning or defunct spacecraft under challenging conditions, such as limited state information, measurement inaccuracies, and erratic target behaviors. Existing approaches often rely on assumptions of full state knowledge or single-integrator dynamics, which are impractical for real-world applications involving second-order spacecraft dynamics.
This work addresses these challenges by developing a distributed state estimation and tracking framework that requires only relative position measurements and operates under partial state information. A novel $\rho$-filter is introduced to reconstruct unknown states using locally available information, and a Lyapunov-based deep neural network adaptive controller is developed that adaptively compensates for uncertainties stemming from unknown spacecraft dynamics. To ensure the collaborative spacecraft regulation problem is well-posed, a trackability condition is defined. A Lyapunov-based stability analysis is provided to ensure exponential convergence of errors in state estimation and spacecraft regulation to a neighborhood of the origin under the trackability condition. The developed method eliminates the need for expensive velocity sensors or extensive pre-training, offering a practical and robust solution for spacecraft servicing in complex, dynamic environments.}


\maketitle

\section{Introduction}

Multi-agent systems play a role in space applications such as distributed
space systems, resilient constellations, and autonomous rendezvous
and docking operations (e.g., \cite{Guffanti2023,Arnas2022,Soderlund2023}).
A rapidly advancing area within this domain is collaborative spacecraft
servicing, which includes on-orbit maintenance, space debris removal,
and swarm-based satellite repositioning (e.g., \cite{Soderlund2023,Mercier2024,Bang2019}).
While current multi-agent spacecraft servicing scenarios often involve
a small number of agents, typically between 2 to 5 spacecraft (e.g.,
\cite{Burch.Moore.ea2016}), future missions are expected to comprise
larger constellations of 10 to 20 or more agents, necessitating the
development of decentralized control strategies that can efficiently
coordinate and adapt to the complexities of these systems (e.g., \cite{Vaughan.Kelley.ea2022}).
The adoption of collaborative and decentralized approaches in these
missions offers several benefits, including enhanced reliability,
reduced costs, and improved safety, as multiple servicing spacecraft
can work together to accomplish complex tasks while minimizing the
risk of single-point failures. These missions often involve multiple
servicing spacecraft performing tasks such as approaching, inspecting,
or rendezvousing with a non-functional or malfunctioning spacecraft
that may be in an uncontrolled tumble or exhibiting anomalous behavior.

The space environment introduces unique challenges for these scenarios.
These include limited knowledge of the defunct spacecraft, reduced
state accuracy due to measurement errors, and partial or unavailable
state measurements (e.g., \cite{Asri2024,Cho2016}). Such constraints
necessitate methods for reconstructing the state of the spacecraft
using incomplete information. Furthermore, the characterization of
defunct spacecraft in complex servicing scenarios is often complicated
by orbital dynamics, poor lighting conditions, unknown physical properties,
and erratic behaviors (e.g., \cite{Kim2012,Morgan2012,Sun2017,Sun.Geng.ea2022}).
In such cases, state estimation alone may prove inadequate, requiring
the additional capability to learn the dynamics of the defunct spacecraft
to predict its behavior and devise effective servicing strategies.

Neural networks (NNs) are frequently employed to approximate the unstructured
uncertainty inherent in spacecraft dynamics, particularly in scenarios
where classical physics-based models may be insufficient or impractical
(e.g., \cite{Zhang2018,Harl2013,Park2023,Silvestrini2020}). In contrast
to white-box approaches, which rely on explicit physical equations,
NNs offer a powerful black-box paradigm that can estimate unknown
dynamics under challenging space weather conditions, including low
Earth orbit, electromagnetic interactions with unknown gravitational
fields, uncertain gravitational perturbations and atmospheric drag.
The ability of NNs to learn complex patterns and relationships in
data makes them well-suited for capturing nonlinearities and uncertainties
that arise in spacecraft systems. Furthermore, recent advances in
deep neural networks (DNNs) have significantly improved function approximation
capabilities, enabling more accurate estimation of uncertain dynamics
(e.g., \cite{LeCun2015}). 

To mitigate the uncertainties inherent in complex systems, traditional
machine learning approaches typically rely on offline training of
NNs using pre-collected datasets. However, this methodology has several
limitations: the required datasets can be difficult to obtain, may
not accurately reflect the operating conditions of the environment,
and fail to adapt online to discrepancies between the pretrained data
and actual system behavior. In contrast, adaptive control offers a
promising alternative by enabling real-time estimation of unknown
model parameters while providing stability guarantees for the system.
While shallow NNs have enabled online adaptation for decades, DNNs
offer enhanced performance, but often require extensive offline pre-training,
which can be challenging in dynamic environments. However, recent
breakthroughs in Lyapunov-based DNNs (Lb-DNNs) overcome the challenges
associated with nonlinear nested uncertain parameters (e.g., \cite{Patil.Le.ea2022,Patil.Le.ea.2022a,Nino.Patil.ea2023,Hart.patil.ea2023,Griffis.Patil.ea2024}),
enabling the construction of analytically-derived update laws based
on Lyapunov-based stability analysis, which provide convergence and
boundedness guarantees while allowing for real-time adaptation without
pre-training requirements. Several results have demonstrated the effectiveness
of adaptive NNs in spacecraft rendezvous applications, including estimating
nonlinear dynamical models in the presence of $J_{2}$ perturbations
and orbit uncertainty estimation (e.g., \cite{Zhang2018,Harl2013}),
and uncertain spacecraft dynamics (e.g., \cite{Park2023,Silvestrini2020}).

The work in \cite{Nino.Patil.ea2025} pioneered the use of Lb-DNNs
for distributed state estimation and tracking in multi-agent systems,
while also introducing the notion of trackability. To illustrate trackability,
consider a network of fixed cameras observing a dynamic object. Each
camera can capture only a subset of the information required to fully
determine the object\textquoteright s state. Independent estimation
by each camera results in incomplete information. By enabling the
cameras to share their individual estimates, the network forms a collaborative
framework that leverages collective information. When generalized
to multi-agent systems, this approach leads to the concept of trackability,
which quantifies the richness of information required for accurate
state estimation in a decentralized system.

For example, in a two-dimensional system with two agents tracking
a single target, if each agent measures only one degree of freedom
and lacks communication with the other, the agents can only align
collinearly with the target. However, successful tracking requires
the agents to share partial measurements to achieve complete state
reconstruction. This concept is particularly relevant for multi-agent
spacecraft servicing, where inter-agent communication can significantly
enhance state estimation and servicing capabilities.

While \cite{Nino.Patil.ea2025} demonstrated exponential state estimation
and tracking within a neighborhood of the object of interest, it relied
on assumptions that are impractical for spacecraft servicing, such
as single-integrator dynamics and full relative state information.
Extending these results to second-order spacecraft dynamics while
assuming full relative velocity information is unrealistic, as relative
velocity sensors are often costly and energy-intensive. Therefore,
there is a need for a distributed observer capable of reconstructing
relative velocities using only relative position measurements.

This work addresses the challenges of spacecraft servicing by eliminating
the restrictive assumptions in \cite{Nino.Patil.ea2025}. The developed
method relies solely on relative position measurements and accommodates
partial state information of the defunct spacecraft. A novel $\rho$-filter
is developed, extending the work in \cite{Nino.Patil.Dixon.2023},
to reconstruct unknown states using locally available information.
Through Lyapunov-based stability analysis, the method guarantees exponential
convergence of the defunct spacecraft\textquoteright s state estimates
and regulates the servicing spacecraft to a neighborhood of the defunct
satellite's state, provided the trackability condition is satisfied.

\section{Notation and Preliminaries}

The study of multi-agent systems involves challenges due to the complexity
arising from multiple interacting agents and nonlinear dynamics. To
address these challenges effectively, several mathematical tools and
notational conventions are introduced in this section. Linear algebra
provides a compact and efficient way to represent agent states, interactions,
and transformations using vectors and matrices, while eigenvalue analysis
supports stability and convergence studies. The networked nature of
multi-agent systems is naturally modeled using graphs, and algebraic
graph theory offers powerful methods to analyze key properties such
as consensus, stability, and connectivity. Finally, deep neural networks
play an increasingly important role in adaptive control of multi-agent
systems, leveraging their ability to approximate complex mappings
and learn from data in uncertain and dynamic environments. These tools
and notations form the foundational elements necessary for modeling,
analysis, and control in the context of multi-agent systems.

Denote by $\mathbf{1}_{n}\in\mathbb{R}^{n}$ the column vector of
length $n>1$ whose entries are all ones. Similarly, $\mathbf{0}_{n}\in\mathbb{R}^{n}$
denotes the column vector of length $n>1$ whose entries are all zeroes.
For $m,n>1$, $\mathbf{0}_{m\times n}\in\mathbb{R}^{m\times n}$ denotes
the $m\times n$ zero matrix. The $p\times p$ identity matrix and
the $p\times1$ column vector of ones are denoted by $I_{p}$ and
$1_{p}$, respectively. Given $M\in\mathbb{Z}_{>0}$, the enumeration
operation $[\cdot]$ is defined as $[M]\triangleq\{1,2,...,M\}$.
Let $n,m\in\mathbb{Z}_{>0}$ with $m>n$. The Euclidean norm of $r\in\mathbb{R}^{n}$
is $\left\Vert r\right\Vert \triangleq\sqrt{r^{\top}r}$. Given a
positive integer $N$ and collection $\{x_{i}\}_{i\in[N]}\subset\mathbb{R}^{n}$,
let $(x_{i})_{i\in[N]}\triangleq[x_{1}^{\top},x_{2}^{\top},...,x_{N}^{\top}]^{\top}\in\mathbb{R}^{nN}$.
Given $H\in\mathbb{R}^{m\times m}$ with columns $\{h_{i}\}_{i\in[n]}\subset\mathbb{R}^{m},\text{vec}(H)=[h_{1}^{\top},h_{2}^{\top},\ldots,h_{n}^{\top}]^{\top}\in\mathbb{R}^{mn}$.
The Frobenius norm is denoted by $\left\Vert \cdot\right\Vert _{F}\triangleq\left\Vert \mathrm{vec}(\cdot)\right\Vert $.
The Kronecker product of $A\in\mathbb{R}^{p\times q}$ and $B\in\mathbb{R}^{m\times n}$
is denoted by $A\otimes B\in\mathbb{R}^{pm\times qn}$. Given any
$A\in\mathbb{R}^{p\times a}$, $B\in\mathbb{R}^{a\times r}$, and
$C\in\mathbb{R}^{r\times s}$, the vectorization operator satisfies
the property $\mathrm{vec}(ABC)=(C^{\top}\otimes A)\mathrm{vec}\left(B\right)$.
Differentiating on both sides with respect to $\mathrm{vec}\left(B\right)$
yields the property
\begin{eqnarray}
\frac{\partial}{\partial\mathrm{vec}\left(B\right)}\mathrm{vec}(ABC) & = & (C^{\top}\otimes A).\label{eq:vec_diff_prop}
\end{eqnarray}

The maximum and minimum eigenvalues of $G=G^{\top}$ are denoted by
$\lambda_{\text{max}}(G)\in\mathbb{R}$ and $\lambda_{\text{min}}(G)\in\mathbb{R}$,
respectively. The spectral norm of $A\in\mathbb{R}^{n\times n}$ is
defined as $\sigma_{\max}\left(A\right)\triangleq\sqrt{\lambda_{\max}\left(A^{\top}A\right)}$.
For $A\in\mathbb{R}^{c\times c}$ and $B\in\mathbb{R}^{d\times d}$,
let the block diagonalization operator be defined as $\text{blkdiag}(A,B)=\left[\begin{array}{cc}
A & \mathbf{0}_{c\times d}\\
\mathbf{0}_{c\times d} & B
\end{array}\right]\in\mathbb{R}^{(c+d)\times(c+d)}$. The right-to-left matrix product operator is represented by $\stackrel{\curvearrowleft}{\prod}$,
i.e., $\stackrel{\curvearrowleft}{\stackrel[p=1]{m}{\prod}}A_{p}=A_{m}\ldots A_{2}A_{1}$
and $\stackrel{\curvearrowleft}{\stackrel[p=a]{m}{\prod}}A_{p}=I$
if $a>m$. 

The space of essentially bounded Lebesgue measurable functions is
denoted by $\mathcal{L}_{\infty}$. For $A\subseteq\mathbb{R}^{n}$
and $B\subseteq\mathbb{R}^{m}$, let ${\tt C}(A,B)$ denote the set
of continuous functions $f:A\to\mathbb{R}^{m}$ such that $f(A)\subseteq B$.
A function with $k$ continuous derivatives is called a $\mathtt{C}^{k}$
function.

\subsection{\label{subsec:Communication-Topology}Algebraic Graph Theory}

Let $\mathcal{G}\triangleq(\mathcal{V},E)$ represent a static and
undirected graph with number of nodes $N\in\mathbb{Z}_{\geq2}$, where
the node set is denoted by $\mathcal{V}\triangleq\left[N\right]$,
and the edge set is denoted by $E\subseteq\mathcal{V}\times\mathcal{V}$.
An edge between nodes $i$ and $k$ belongs to the edge set $(\text{i.e., }(i,k)\in E)$
if and only if node $i$ can send information to node $k$. Since
the graph $\mathcal{G}$ is undirected, $(i,k)\in E$ if and only
if $(k,i)\in E$. An undirected graph is connected whenever there
exists a sequence of edges in $E$ linking any two distinct nodes.
The neighborhood set of node $i$ is $\mathcal{N}_{i}\triangleq\{k\in\mathcal{V}\setminus\{i\}:(k,i)\in E\}$.
Let $\mathcal{A}\triangleq\left[a_{ik}\right]\in\mathbb{R}^{N\times N}$
be the adjacency matrix of $\mathcal{G}$, where $a_{ik}=1$ if $(k,i)\in E$
and $a_{ik}=0$ otherwise. Within this work, no self-loops are considered.
Therefore, $a_{ii}\triangleq0$ for all $i\in\mathcal{V}$. The degree
matrix of $\mathcal{G}$ is $\mathcal{D}\triangleq\text{diag}(\mathcal{A}\cdot1_{N})\in\mathbb{R}^{N\times N}$.
Using the degree and adjacency matrices, the Laplacian matrix of the
graph $\mathcal{G}$ is $\mathcal{L}\triangleq\mathcal{D}-\mathcal{A}$.

\subsection{Deep Neural Network Model}

Let $\kappa\in\mathbb{R}^{L_{0}}$ denote a DNN input, and $\theta\in\mathbb{R}^{p}$
denote the vector of DNN parameters (i.e., weights and bias terms).
A fully-connected feedforward DNN $\Phi(\kappa,\theta)$ with $k\in\mathbb{Z}_{>0}$
hidden layers and output size $L_{k+1}\in\mathbb{Z}_{>0}$ is defined
using a recursive relation $\varphi_{j}\in\mathbb{R}^{L_{j+1}}$ modeled
as \cite{Patil.Le.ea2022}
\begin{eqnarray}
\varphi_{j} & \triangleq & \begin{cases}
V_{j+1}^{\top}\kappa_{a}, & j=0,\\
V_{j+1}^{\top}\phi_{j}\left(\varphi_{j-1}\right) & j\in\left\{ 1,\ldots,k\right\} ,
\end{cases}\label{eq:DNN}
\end{eqnarray}
where $\Phi(\kappa,\theta)=\varphi_{k}$, $\kappa_{a}\triangleq\left[\kappa^{\top},1\right]^{\top}$
denotes the augmented input that accounts for the bias terms, $L_{j}\in\mathbb{Z}_{>0}$
denotes the a in the $j^{\textrm{th}}$ layer with $L_{j}^{a}\triangleq L_{j}+1$,
and $V_{j+1}\in\mathbb{R}^{L_{j}^{a}\times L_{j+1}}$ denotes the
matrix of weights and biases, for all $j\in\left\{ 0,\ldots,k\right\} $.

The vector of activation functions is denoted by $\phi_{j}:\mathbb{R}^{L_{j}}\to\mathbb{R}^{L_{j}^{a}}$
for all $j\in\left\{ 1,\ldots,k\right\} $. The vector of activation
functions can be composed of various activation functions, and hence,
may be represented as $\phi_{j}=\left[\varsigma_{1},\ldots,\varsigma_{L_{j}},1\right]^{\top}$
for all $j\in\left\{ 1,\ldots,k\right\} $, where $\varsigma_{j}:\mathbb{R}\to\mathbb{R}$
for all $j\in\left\{ 1,\ldots,L_{j}\right\} $ denotes a bounded $\mathtt{C}^{1}$
activation function, where 1 accounts for the bias term. For the DNN
architecture in (\ref{eq:DNN}), the vector of DNN weights is $\theta\triangleq\left[\mathrm{vec}(V_{1})^{\top},\ldots,\mathrm{vec}(V_{k})^{\top}\right]^{\top}$
with size $p=\sum_{j=0}^{k}L_{j}^{a}L_{j+1}$.

Consider $y_{j}\in\mathbb{R}^{L_{j}}$ where $y_{j}=\left[y_{1},\ldots,y_{L_{j}}\right]$
with $y_{i}\in\mathbb{R}$ for all $i\in\left\{ 1,\ldots,L_{j}\right\} $.
The Jacobian $\frac{\partial\phi_{j}}{\partial y_{j}}:\mathbb{R}^{L_{j}}\to\mathbb{R}^{L_{j}^{a}\times L_{j}}$
of the activation function vector at the $j^{\mathrm{th}}$ layer
is given by $\left[\varsigma_{1}^{\prime}(y_{1})\mathbf{e}_{1},\ldots,\varsigma_{L_{j}}^{\prime}(y_{L_{j}})\mathbf{e}_{L_{j}},\mathbf{0}_{L_{j}}\right]^{\top}$$\in\mathbb{R}^{L_{j}^{a}\times L_{j}}$,
where $\varsigma_{j}^{\prime}$ denotes the derivative of $\varsigma_{j}$
with respect to its argument for $j\in\left\{ 1,\ldots,L_{j}\right\} $,
$\mathbf{e}_{i}$ is the $i^{\text{th}}$ standard basis vector in
$\mathbb{R}^{L_{j}}$, and $\mathbf{0}_{L_{j}}$ is the zero vector
in $\mathbb{R}_{L_{j}}$.

Let the gradient of the DNN with respect to the weights be denoted
by $\nabla_{\theta}\Phi(\kappa,\theta)\triangleq\frac{\partial}{\partial\theta}\Phi(\kappa,\theta)$,
which can be represented as $\nabla_{\theta}\Phi(\kappa,\theta)=\left[\frac{\partial}{\partial{\rm vec}(V_{1})}\Phi(\kappa,\theta),\ldots,\frac{\partial}{\partial{\rm vec}(V_{k+1})}\Phi(\kappa,\theta)\right]$$\in\mathbb{R}^{L_{k+1}\times p}$,
where $\frac{\partial}{\partial{\rm vec}(V_{j})}\Phi(\kappa,\theta)$$\in\mathbb{R}^{L_{k+1}\times L_{j-1}^{a}L_{j}}$
for all $j\in\left\{ 1,\ldots,k+1\right\} $. Using (\ref{eq:DNN})
and the property of the vectorization operator in (\ref{eq:vec_diff_prop})
yields
\begin{align}
\nabla_{\theta}\Phi(\kappa,\theta) & =\left(\stackrel{\curvearrowleft}{\stackrel[\ell=j+1]{k}{\prod}}V_{\ell+1}^{\top}\frac{\partial\phi_{\ell}}{\partial\varphi_{\ell-1}}\right)\left(I_{L_{j+1}}\otimes\varrho_{j}\right),\label{eq:DNN Gradient}
\end{align}
for $j\in\left\{ 0,\ldots,k\right\} $, where $\varrho_{j}=\kappa_{a}^{\top}$
if $j=0$ and $\varrho_{j}=\phi_{j}^{\top}\left(\varphi_{j-1}\right)$
if $j\in\left\{ 1,\ldots,k\right\} $.\footnote{The following control and adaptation law development can be generalized
for any neural network architecture $\Phi$ with a corresponding Jacobian
$\nabla_{\theta}\Phi$. The reader is referred to \cite{Griffis.Patil.ea23_2}
and \cite{Patil.Le.ea.2022a} for extending the subsequent development
to LSTMs and ResNets, respectively.}

\section{Problem Formulation}

\subsection{System Dynamics}

Consider a servicer spacecraft indexed by $i$ and a defunct spacecraft
indexed by $0$. The relative motion of servicer spacecraft with respect
to the defunct spacecraft, which is assumed to follow an elliptical
Keplerian orbit, can be described by the following system of differential
equations (see \cite{Schaub.Junkins2003})
\begin{align*}
\ddot{q}_{i}^{x} & =2\tau\dot{q}_{i}^{y}+\dot{\tau}q_{i}^{y}+\tau^{2}q_{i}^{x}+\frac{\mu}{{\tt r}^{2}}-\frac{\mu({\tt r}+q_{i}^{x})}{\left(\left({\tt r}+q_{i}^{x}\right)^{2}+(q_{i}^{y})^{2}+(q_{i}^{z})^{2}\right)^{\frac{3}{2}}}+\mathbf{P}_{i}^{x}(q_{i}^{x})\\
 & +\mathbf{D}_{i}^{x}(\dot{q}_{i}^{x})+u_{i}^{x},\\
\ddot{q}_{i}^{y} & =-2\tau\dot{q}_{i}^{x}-\dot{\tau}q_{i}^{x}+\tau^{2}q_{i}^{y}-\frac{\mu q_{i}^{y}}{\left(\left({\tt r}+q_{i}^{x}\right)^{2}+(q_{i}^{y})^{2}+(q_{i}^{z})^{2}\right)^{\frac{3}{2}}}+\mathbf{P}_{i}^{y}(q_{i}^{y})+\mathbf{D}_{i}^{y}(\dot{q}_{i}^{y})+u_{i}^{y},\\
\ddot{q}_{i}^{z} & =-\frac{\mu q_{i}^{z}}{\left(\left({\tt r}+q_{i}^{x}\right)^{2}+(q_{i}^{y})^{2}+(q_{i}^{z})^{2}\right)^{\frac{3}{2}}}+\mathbf{P}_{i}^{z}(q_{i}^{z})+\mathbf{D}_{i}^{z}(\dot{q}_{i}^{z})+u_{i}^{z},\\
\ddot{{\tt r}} & ={\tt r}\tau^{2}-\frac{\mu}{{\tt r}^{2}},\\
\dot{\tau} & =-\frac{2\dot{{\tt r}}\tau}{{\tt r}},
\end{align*}
where, $q_{i}^{x},q_{i}^{y},q_{i}^{z}\in\mathbb{R}$ represent the
radial, in-track, and cross-track rectangular coordinates of the servicer
spacecraft relative to the defunct spacecraft, $u_{i}^{x},u_{i}^{y},u_{i}^{z}\in\mathbb{R}$
denote the control accelerations of the servicer spacecraft, ${\tt r},\tau\in\mathbb{R}$
are the orbital radius and angular velocity of the defunct spacecraft,
respectively, and $\mu$ is the gravitational parameter of the central
body. The additional forces influencing the servicer's motion include
gravitational perturbations $\mathbf{P}_{i}^{x}(q_{i}^{x}),\mathbf{P}_{i}^{y}(q_{i}^{y}),\mathbf{P}_{i}^{z}(q_{i}^{z})\in\mathbb{R}$,
which model deviations from an ideal central gravitational field,
such as oblateness ($J_{2}$) or third-body effects (see \cite[Equation 3]{Vess.Starin2003}),
and atmospheric drag forces $\mathbf{D}_{i}^{x}(\dot{q}_{i}^{x}),\mathbf{D}_{i}^{y}(\dot{q}_{i}^{y}),\mathbf{D}_{i}^{z}(\dot{q}_{i}^{z})\in\mathbb{R}$
(see \cite[Equation 1]{MostazaPrieto.Graziano.ea2014}).

For spacecraft using a radar system for rendezvous navigation, the
transformation of variables
\begin{align*}
q_{i}^{x} & =\sigma_{i}\cos\phi_{i}\cos\gamma_{i},\\
q_{i}^{y} & =\sigma_{i}\cos\phi_{i}\sin\gamma_{i},\\
q_{i}^{z} & =\sigma_{i}\sin\phi_{i},
\end{align*}
are used (see \cite{Eggleston1961}), where $\sigma_{i}$ is the range
between the servicer spacecraft $i$ and the defunct spacecraft, $\gamma_{i}$
is the azimuth angle, and $\phi_{i}$ is the relative elevation angle.
After the substitution of this transformation, the relative motion
dynamics are given by
\begin{equation}
\ddot{q}_{i}=f_{i}(q_{i},\dot{q}_{i})+g(q_{i})u_{i}+\omega(\tau,\dot{\tau},q_{i},\dot{q}_{i})\label{eq:Dynamics}
\end{equation}
where $q_{i}=\left[\begin{array}{ccc}
\sigma_{i} & \gamma_{i} & \phi_{i}\end{array}\right]^{\top}\in\mathbb{R}^{3}$, $u_{i}=\left[\begin{array}{ccc}
u_{i}^{x} & u_{i}^{y} & u_{i}^{z}\end{array}\right]^{\top}\in\mathbb{R}^{3}$, $f_{i}=\left[\begin{array}{ccc}
f^{\sigma} & f^{\gamma} & f^{\phi}\end{array}\right]\in\mathbb{R}^{3}$, $\omega=\left[\begin{array}{ccc}
\omega^{\sigma} & \omega^{\gamma} & \omega^{\phi}\end{array}\right]\in\mathbb{R}^{3}$ where
\begin{align}
f_{i}^{\sigma}(q_{i},\dot{q}_{i}) & =\sigma_{i}\dot{\phi}^{2}-\frac{\mu({\tt r}\cos\gamma_{i}\cos\phi_{i}+\sigma_{i})}{\left({\tt r}^{2}+\sigma_{i}^{2}+2{\tt r}\sigma_{i}\cos\gamma_{i}\cos\phi_{i}\right)^{\frac{3}{2}}}+\frac{\mu}{{\tt r}^{2}}\cos\gamma_{i}\cos\phi_{i}\nonumber \\
 & +\mathbf{P}_{i}^{x}(q_{i}^{x})+\mathbf{D}_{i}^{x}(\dot{q}_{i}^{x}),\nonumber \\
f_{i}^{\gamma}(q_{i},\dot{q}_{i}) & =\frac{\mu({\tt r}\sin\gamma_{i}\sec\phi_{i})}{\sigma_{i}\left({\tt r}^{2}+\sigma_{i}^{2}+2{\tt r}\sigma_{i}\cos\gamma_{i}\cos\phi_{i}\right)^{\frac{3}{2}}}+\frac{\mu}{{\tt r}^{2}\sigma_{i}}\sin\gamma_{i}\sec\phi_{i}+\mathbf{P}_{i}^{y}(q_{i}^{y})+\mathbf{D}_{i}^{y}(\dot{q}_{i}^{y}),\nonumber \\
f_{i}^{\sigma}(q_{i},\dot{q}_{i}) & =-\frac{2\dot{\sigma}_{i}\dot{\phi}_{i}}{\sigma_{i}}-\frac{\mu}{{\tt r}^{2}\sigma_{i}}\cos\gamma_{i}\sin\phi_{i}+\frac{\mu{\tt r}\cos\gamma_{i}\sin\phi_{i}}{\sigma_{i}\left({\tt r}^{2}+\sigma_{i}^{2}+2{\tt r}\sigma_{i}\cos\gamma_{i}\cos\phi_{i}\right)^{\frac{3}{2}}}\nonumber \\
 & +\mathbf{P}_{i}^{z}(q_{i}^{z})+\mathbf{D}_{i}^{z}(\dot{q}_{i}^{z}),\label{eq:drift}
\end{align}
\begin{equation}
g(q_{i})=\left[\begin{array}{ccc}
1 & 0 & 0\\
0 & \frac{1}{\sigma_{i}} & 0\\
0 & 0 & \frac{1}{\sigma_{i}}
\end{array}\right],\label{eq:controlEffec}
\end{equation}
and 
\begin{align}
\omega^{\sigma}(\tau,q_{i},\dot{q}_{i}) & =(\tau^{2}+2\tau\dot{\gamma}+\dot{\gamma}^{2})\sigma\cos^{2}\phi,\nonumber \\
\omega^{\gamma}(\tau,\dot{\tau},q_{i},\dot{q}_{i}) & =2(\tau+\dot{\gamma})\dot{\phi}\tan\phi-\dot{\tau}-2(\tau+\dot{\gamma})\frac{\dot{\sigma}}{\sigma},\nonumber \\
\omega^{\sigma}(\tau,\dot{q}_{i}) & =-\frac{1}{2}(\tau+\dot{\gamma})^{2}.\label{eq:Disturbances}
\end{align}

\subsection{Multi-Spacecraft System Model}

Consider a multi-spacecraft system of servicing spacecraft consisting
of $N$ agents indexed by $i\in\mathcal{V}$, and a single defunct
spacecraft indexed by $\{0\}$, with $\overline{\mathcal{V}}\triangleq\mathcal{V}\cup\{0\}$.
The dynamics for spacecraft $i\in\mathcal{V}$ are given by (\ref{eq:Dynamics}).
Based on the structure of $\omega$ as defined by (\ref{eq:Disturbances}),
there exist known constants $\overline{\omega},\overline{\dot{\omega}}\in\mathbb{R}_{>0}$
such that $\left\Vert \omega(\tau,\dot{\tau},q_{i},\dot{q}_{i})\right\Vert \leq\overline{\omega}$
and $\left\Vert \dot{\omega}(\tau,\dot{\tau},q_{i},\dot{q}_{i})\right\Vert \leq\overline{\dot{\omega}}$
for all $t\in\left[0,\infty\right)$.

The dynamics for the defunct spacecraft is given by
\begin{equation}
\ddot{q}_{0}=f_{0}(q_{0},\dot{q}_{0}),\label{eq:Target Dynamics}
\end{equation}
where $q_{0},\dot{q}_{0},\ddot{q}_{0}\in\mathbb{R}^{3}$ denote the
defunct spacecrafts unknown generalized position, velocity, and acceleration,
respectively, and the function $f_{0}:\mathbb{R}^{3}\times\mathbb{R}^{3}\to\mathbb{R}^{3}$
is unknown and of class $\mathtt{C}^{1}$. In practice, the target
spacecraft's motion is often constrained by its initial conditions,
orbital mechanics, or other environmental factors, which limits its
state variations. The following assumption formally captures these
conditions. 
\begin{assumption}
\label{targetbounds}There exist known constants $\overline{q}_{0},\overline{\dot{q}}_{0}\in\mathbb{R}_{>0}$
such that $\left\Vert q_{0}(t)\right\Vert \leq\overline{q}_{0}$ and
$\left\Vert \dot{q}_{0}(t)\right\Vert \leq\overline{\dot{q}}_{0}$
for all $t\in\left[0,\infty\right)$.
\end{assumption}

\subsection{Control Objective}

Each spacecraft \textit{$i\in\mathcal{V}$} can measure the relative
position $d_{i,j}\in\mathbb{R}^{3}$ between itself and its neighbors
$j\in\mathcal{N}_{i}$, defined as
\begin{equation}
d_{i,j}\triangleq q_{j}-q_{i}.\label{eq:relative position}
\end{equation}
Unique to the defunct spacecraft, each spacecraft $i\in\mathcal{V}$
can measure the partial relative position between itself and the defunct
spacecraft, given by
\begin{equation}
y_{i}\triangleq C_{i}\left(q_{0}-q_{i}\right),\label{eq:measurement model}
\end{equation}
where $y_{i}\in\mathbb{R}^{m_{i}}$ and $C_{i}\in\mathbb{R}^{m_{i}\times3}$
represents the output matrix of agent $i$, characterizing the agent\textquoteright s
heterogeneous sensing capabilities. Since each spacecraft is typically
equipped with a specific suite of sensors, such as cameras or LIDAR,
it is reasonable to assume that each spacecraft has knowledge of its
own sensor configuration and capabilities. Formally, each agent $i\in\mathcal{V}$
knows its own output matrix $C_{i}$.

The primary objective is to design a distributed controller for each
servicer spacecraft $i\in\mathcal{V}$, that guides the servicer spacecraft
towards the defunct spacecraft using only the partial relative measurement
model. Since relative velocity measurements are unavailable, a secondary
objective is to develop a decentralized observer that can estimate
the relative velocities using only locally available information from
each spacecraft. Additionally, because the defunct spacecrafts state
is unknown, a tertiary objective is to design a distributed system
identifier that reconstructs the defunct spacecrafts unknown state
while simultaneously using online learning techniques to approximate
its dynamics.

To quantify the servicing objective, define the tracking error $e_{i}\in\mathbb{R}^{3}$
of agent $i\in\mathcal{V}$ as
\begin{equation}
e_{i}\triangleq q_{0}-q_{i}.\label{eq:tracking error}
\end{equation}
Furthermore, define the relative position error $\eta_{i}\in\mathbb{R}^{3}$
as
\begin{equation}
\eta_{i}\triangleq\sum_{j\in\mathcal{N}_{i}}d_{i,j}+b_{i}C_{i}^{\top}y_{i},\label{eq: Implementable Position}
\end{equation}
where $b_{i}\in\left\{ 0,1\right\} $ denotes a binary indicator of
spacecraft $i$'s ability to sense the defunct spacecraft, for all
$i\in\mathcal{V}$. 

Using (\ref{eq:relative position}) and (\ref{eq:tracking error}),
(\ref{eq: Implementable Position}) is expressed in an equivalent
analytical form as
\begin{equation}
\eta_{i}=\left(b_{i}C_{i}^{\top}C_{i}e_{i}-\sum_{j\in\mathcal{N}_{i}}\left(e_{j}-e_{i}\right)\right),\label{eq:relative position sychronization}
\end{equation}
for all $i\in\mathcal{V}$.

In spacecraft servicing scenarios, reliable communication among spacecraft
enables efficient coordination and execution of tasks. The communication
topology of the network is often designed to be connected, ensuring
that information can be exchanged between all spacecraft. This connectivity
enables the spacecraft to share resources and adapt to changing mission
requirements. The subsequent assumption provides a mathematical representation
of these conditions.
\begin{assumption}
\label{thm:connected} The graph $\mathcal{G}$ is connected, and
there exists at least one $b_{i}=1$ for some $i\in\mathcal{V}$.
\end{assumption}

\section{Control Design}

Define the filtered tracking error $r_{i}\in\mathbb{R}^{3}$ as
\begin{equation}
r_{i}=\dot{e}_{i}+k_{1}e_{i},\label{eq:Filtered Tracking Error}
\end{equation}
where $k_{1}\in\mathbb{R}_{>0}$ is a user-defined constant, for all
$i\in\mathcal{V}$. Let $\hat{\eta}_{i}\in\mathbb{R}^{3}$ and $\hat{\zeta}_{i}\in\mathbb{R}^{3}$
denote the relative position error and relative velocity error estimates,
respectively. The corresponding relative position estimation error
$\tilde{\eta}_{i}\in\mathbb{R}^{3}$ and relative velocity estimation
error $\tilde{\zeta}_{i}\in\mathbb{R}^{3}$ are defined as
\begin{align}
\tilde{\eta}_{i} & =\eta_{i}-\hat{\eta}_{i},\label{eq:position estimation error}\\
\tilde{\zeta}_{i} & =\zeta_{i}-\hat{\zeta}_{i},\label{eq:velocity estimation error}
\end{align}
where $\zeta_{i}\triangleq\dot{\eta}_{i}$, for all $i\in\mathcal{V}$.
Taking the second time-derivative of (\ref{eq:tracking error}), substituting
(\ref{eq:Dynamics}) and (\ref{eq:Target Dynamics}) into the resulting
expression, and adding and subtracting by $f_{0}(\eta_{i},\hat{\zeta}_{i})-f_{i}(\eta_{i},\hat{\zeta}_{i})$
yields
\begin{align}
\ddot{e}_{i} & =f_{0}(\eta_{i},\hat{\zeta}_{i})-f_{i}(\eta_{i},\hat{\zeta}_{i})-g(q_{i})u_{i}(t)-\omega(\tau,\dot{\tau},q_{i},\dot{q}_{i})+\tilde{f}_{i}(q_{0},\dot{q}_{0},q_{i},\dot{q}_{i},\eta_{i},\hat{\zeta}_{i}),\label{eq:tracking error acceleration}
\end{align}
where $\tilde{f}_{i}(q_{0},\dot{q}_{0},q_{i},\dot{q}_{i},\eta_{i},\hat{\zeta}_{i})\triangleq f_{0}(q_{0},\dot{q}_{0})-f_{0}(\eta_{i},\hat{\zeta}_{i})+f_{i}(\eta_{i},\hat{\zeta}_{i})-f_{i}(q_{i},\dot{q}_{i})\in\mathbb{R}^{3}$.
Substituting (\ref{eq:tracking error acceleration}) into the time
derivative of (\ref{eq:Filtered Tracking Error}) yields
\begin{align}
\dot{r}_{i} & =f_{0}(\eta_{i},\hat{\zeta}_{i})-f_{i}(\eta_{i},\hat{\zeta}_{i})-g(q_{i})u_{i}(t)-\omega(\tau,\dot{\tau},q_{i},\dot{q}_{i})+\tilde{f}_{i}(q_{0},\dot{q}_{0},q_{i},\dot{q}_{i},\eta_{i},\hat{\zeta}_{i})+k_{1}\dot{e}_{i}.\label{eq:Filtered Tracking Error Derivative-1}
\end{align}

\subsection{Lyapunov-based Deep Neural Network Function Approximation}

Traditional physics-based models can be challenging to formulate in
a manner that accurately captures the complex dynamics of spacecraft
in servicing scenarios, due to simplifying assumptions and uncertainties
such as varying mass properties and environmental conditions. DNNs
offer a promising alternative that can learn complex patterns from
data without requiring explicit knowledge of the underlying physics.
This function approximation capability motivates the use of DNN-based
function approximation to develop more accurate and robust models
of spacecraft dynamics. The universal function approximation property
is assumed to hold over the compact set $\Omega\subset\mathbb{R}^{6}$,
defined as
\begin{equation}
\Omega\triangleq\left\{ \iota\in\mathbb{R}^{6}:\left\Vert \iota\right\Vert \leq\Upsilon\right\} ,\label{eq:Omega}
\end{equation}
where $\Upsilon\in\mathbb{R}_{>0}$ is a positive constant, for all
$k\in\overline{\mathcal{V}}$.\footnote{The domain $\mathbb{R}^{3}\times\mathbb{R}^{3}$ is identified with
$\mathbb{R}^{6}$ by treating $\left[\begin{array}{cc}
\mathtt{x} & \mathtt{y}\end{array}\right]\in\mathbb{R}^{3}\times\mathbb{R}^{3}$ as $\iota=\left[\begin{array}{cccccc}
\mathtt{x}_{1} & \cdots & \mathtt{x}_{3} & \mathtt{y}_{1} & \cdots & \mathtt{y}_{3}\end{array}\right]\in\mathbb{R}^{6}$.}

Let $\kappa_{i}\triangleq\left[\begin{array}{cc}
\eta_{i}^{\top} & \hat{\zeta}_{i}^{\top}\end{array}\right]^{\top}\in\mathbb{R}^{6}$ and define $h_{i}:\mathbb{R}^{6}\to\mathbb{R}^{3}$ as $h_{i}(\kappa_{i})\triangleq f_{0}(\eta_{i},\hat{\zeta}_{i})-f_{i}(\eta_{i},\hat{\zeta}_{i})$,
for all $i\in\mathcal{V}$. Prescribe $\overline{\varepsilon}>0$
and note that for all $i\in\mathcal{V}$, $h_{i}\in{\tt C}\left(\Omega,\mathbb{R}^{3}\right)$.
Then, by \cite[Theorem 3.2]{Kidger2020}, there exists an Lb-DNN such
that $\sup_{\kappa_{i}\in\Omega}\left\Vert \Phi_{i}(\kappa_{i},\theta_{i}^{\ast})-h_{i}(\kappa_{i})\right\Vert <\overline{\varepsilon}$,
for all $i\in\mathcal{V}$. Therefore, each agent $i\in\mathcal{V}$
can model the unknown function $h_{i}(\kappa_{i})$ using an Lb-DNN
as

\begin{equation}
h_{i}(\kappa_{i})=\Phi_{i}(\kappa_{i},\theta_{i}^{\ast})+\varepsilon_{i}(\kappa_{i}),\label{eq:LbDNN}
\end{equation}
where $\theta_{i}^{\ast}\in\mathbb{R}^{p}$ are the ideal weights,
$\Phi_{i}:\mathbb{R}^{6}\times\mathbb{R}^{p}\to\mathbb{R}^{3}$, and
$\varepsilon_{i}:\mathbb{R}^{6}\to\mathbb{R}^{3}$ is an unknown function
representing the reconstruction error that is bounded as $\sup_{\kappa_{i}\in\Omega}\left\Vert \varepsilon_{i}(\kappa_{i})\right\Vert <\overline{\varepsilon}$.

The Lb-DNN described in (\ref{eq:LbDNN}) is inherently nonlinear
with respect to its weights. To address the nonlinearity, a first-order
Taylor approximation for the Lb-DNN is applied. To quantify the approximation,
the parameter estimation error $\widetilde{\theta}_{i}\in\mathbb{R}^{p}$
is defined as
\begin{equation}
\widetilde{\theta}_{i}=\theta_{i}^{\ast}-\widehat{\theta}_{i},\label{eq: parameter estimation error}
\end{equation}
for all $i\in\mathcal{V}$, where $\widehat{\theta}_{i}\in\mathbb{R}^{p}$
represents the weight estimates. The first-order Taylor approximation
of $\Phi_{i}(\kappa_{i},\theta_{i}^{\ast})$ evaluated at $\hat{\theta}_{i}$
is given as
\begin{equation}
\Phi_{i}(\kappa_{i},\theta_{i}^{\ast})=\Phi_{i}(\kappa_{i},\hat{\theta}_{i})+\nabla_{\hat{\theta}_{i}}\Phi_{i}(\kappa_{i},\hat{\theta}_{i})\widetilde{\theta}_{i}+R_{1,i}(\widetilde{\theta}_{i},\kappa_{i}),\label{eq:TaylorApproximation}
\end{equation}
where $\nabla_{\widehat{\theta}_{i}}\Phi_{i}:\mathbb{R}^{6}\times\mathbb{R}^{p}\to\mathbb{R}^{3\times p}$
is the Jacobian of $\Phi_{i}$ with respect to $\hat{\theta}_{i}$,
and $R_{1,i}:\mathbb{R}^{p}\times\mathbb{R}^{6}\to\mathbb{R}^{3}$
is the first Lagrange remainder, which accounts for the error introduced
by truncating the Taylor approximation after the first-order term,
for all $i\in\mathcal{V}$.

Using (\ref{eq:LbDNN}) and (\ref{eq:TaylorApproximation}), (\ref{eq:Filtered Tracking Error Derivative-1})
is rewritten as
\begin{align}
\dot{r}_{i} & =\Phi_{i}(\kappa_{i},\hat{\theta}_{i})+\nabla_{\hat{\theta}_{i}}\Phi_{i}(\kappa_{i},\hat{\theta}_{i})\widetilde{\theta}_{i}-g(q_{i})u_{i}(t)-\omega(\tau,\dot{\tau},q_{i},\dot{q}_{i})+k_{1}\dot{e}_{i}\nonumber \\
 & +\tilde{f}_{i}(q_{0},\dot{q}_{0},q_{i},\dot{q}_{i},\kappa_{i})+\Delta_{i}(\widetilde{\theta}_{i},\kappa_{i}),\label{eq:Filtered Tracking Error Derivative-2}
\end{align}
where $\Delta_{i}(\widetilde{\theta}_{i},\kappa_{i})\triangleq R_{1,i}(\widetilde{\theta}_{i},\kappa_{i})+\varepsilon_{i}(\kappa_{i})\in\mathbb{R}^{3}$.
By \cite[Theorem 4.7]{Lax2017}, the remainder term $R_{1,i}(\widetilde{\theta}_{i},\kappa_{i})$
can be expressed as $R_{1,i}(\widetilde{\theta}_{i},\kappa_{i})=\frac{1}{2}\tilde{\theta}_{i}^{\top}\nabla_{\hat{\theta}_{i}}^{2}\Phi_{i}\left(\kappa_{i},\hat{\theta}_{i}+\gamma_{i}\tilde{\theta}_{i}\right)\tilde{\theta}_{i}$,
where $\nabla_{\widehat{\theta}_{i}}^{2}\Phi_{i}:\mathbb{R}^{6}\times\mathbb{R}^{p}\to\mathbb{R}^{6\times p\times p}$
is the Hessian of $\Phi_{i}$ with respect to $\hat{\theta}_{i}$,
and $\gamma_{i}\in\left[0,1\right]$, for all $i\in\mathcal{V}$.
Consequently, there exists some constant $M_{i}\in\mathbb{R}_{>0}$
such that $\sigma_{\max}\left(\nabla_{\widehat{\theta}_{i}}^{2}\Phi_{i}\left(\kappa_{i},\hat{\theta}_{i}+\gamma_{i}\tilde{\theta}_{i}\right)\right)\leq M_{i}$,
which, by \cite[Theorem 8.8]{Zhang2011}, yields $\left\Vert R_{1,i}(\widetilde{\theta}_{i},\kappa_{i})\right\Vert \leq\frac{M_{i}}{2}\left\Vert \tilde{\theta}_{i}\right\Vert ^{2}$,
given bounded $\kappa_{i}$, for all $i\in\mathcal{V}$.

To facilitate the subsequent development, the following assumption
is made.
\begin{assumption}
\label{Asp: DNN Bounds}\cite[Assumption 1]{Lewis1996b} There exists
$\overline{\theta}\in\mathbb{R}_{>0}$ such that the unknown ideal
weights can be bounded as $\max_{i\in\mathcal{V}}\left\{ \left\Vert \theta_{i}^{\ast}\right\Vert \right\} \leq\overline{\theta}$.
\end{assumption}
Assumption \ref{Asp: DNN Bounds} is reasonable since, in practice,
the user can select $\overline{\theta}$ a priori, and subsequently
prescribe $\overline{\varepsilon}$ using a conservative estimate
whose feasibility can be verified using heuristic search methods,
e.g., Monte Carlo search. Alternatively, adaptive bound estimation
techniques (e.g., \cite{Fan2018}) can be used to estimate $\overline{\theta}$.

\subsection{Distributed Observer-based Control Design}

The controller, filter, observer, and adaptation law are designed
to satisfy a set of conditions that ensure the stability of the closed-loop
system. Specifically, the design aims to cancel out cross-coupled
terms in the Lyapunov function derivative, while bounding or utilizing
the remaining terms to achieve negative definiteness. To this end,
the control design incorporates several key features, including the
introduction of auxiliary variables, such as a filtered estimation
error, and the selection of user-defined constants. The subsequent
stability analysis will demonstrate that these design elements establish
the boundedness and convergence properties of the system and will
provide a rigorous justification for the specific structure of the
control input, observer, and adaptation law. 

To facilitate the distributed observer design, a filtered estimation
error $\tilde{r}_{i}\in\mathbb{R}^{3}$ is defined as
\begin{equation}
\tilde{r}_{i}\triangleq\dot{\tilde{\eta}}_{i}+k_{3}\tilde{\eta}_{i}+\rho_{i},\label{eq:Filtered Estimation Error}
\end{equation}
for all $i\in\mathcal{V}$, where $k_{3}\in\mathbb{R}_{>0}$ is a
user-defined constant, and $\rho_{i}\in\mathbb{R}^{3}$ is designed
as
\begin{align}
\rho_{i}(t) & =-(k_{3}+k_{4})\tilde{\eta}_{i}(t)\nonumber \\
 & +\left(1-k_{3}^{2}-k_{3}k_{4}\right)\int_{t_{0}}^{t}\tilde{\eta}_{i}(\tau){\rm d}\tau-\left(k_{3}+k_{4}+k_{5}\right)\int_{t_{0}}^{t}\rho_{i}(\tau){\rm d}\tau,\nonumber \\
\rho_{i}(0) & =\mathbf{0}_{3},\label{eq:rho implemented}
\end{align}
where $k_{4},k_{5}\in\mathbb{R}_{>0}$ are user-defined constants.
The distributed observer is designed as
\begin{align}
\dot{\hat{\eta}}_{i} & =\hat{\zeta}_{i},\nonumber \\
\dot{\hat{\zeta}}_{i} & =\left(\sum_{j\in\mathcal{N}_{i}}\left(g_{j}u_{j}-g_{i}u_{i}\right)-b_{i}C_{i}^{\top}C_{i}g_{i}u_{i}\right)-(k_{3}^{2}-2)\tilde{\eta}_{i}-\left(2k_{3}+k_{4}+k_{5}\right)\rho_{i},\nonumber \\
\hat{\eta}_{i}(0) & =\mathbf{0}_{3},\nonumber \\
\hat{\zeta}_{i}(0) & =\mathbf{0}_{3},\label{eq: Distributed Observer}
\end{align}
for all $i\in\mathcal{V}$. Based on the subsequent stability analysis,
the control input is designed as
\begin{align}
u_{i} & =g_{i}^{-1}\left(\Phi_{i}(\kappa_{i},\hat{\theta}_{i})+k_{2}\left(k_{1}\eta_{i}+\hat{\zeta}_{i}-k_{3}\tilde{\eta}_{i}-\rho_{i}\right)\right),\label{eq: controller-1}
\end{align}
where $k_{2}\in\mathbb{R}_{>0}$ is a user-defined constant and $g_{i}^{-1}$
is guaranteed to exist by (\ref{eq:controlEffec}), for all $i\in\mathcal{V}$.
Similarly, the adaptation law for the Lb-DNN is designed as
\begin{align}
\dot{\hat{\theta}}_{i} & =\text{proj}\left(\Theta_{i},\hat{\theta}_{i},\overline{\theta}\right),\label{eq:adaptive update law}
\end{align}
where
\begin{align*}
\Theta_{i} & \triangleq\Gamma_{i}\left(\nabla_{\hat{\theta}_{i}}\Phi_{i}(\kappa_{i},\hat{\theta}_{i})\left(\hat{\zeta}_{i}+k_{1}\eta_{i}\right)-k_{6}\left(\hat{\theta}_{i}-\sum_{j\in\mathcal{N}_{i}}\left(\hat{\theta}_{j}-\hat{\theta}_{i}\right)\right)\right),
\end{align*}
for all $i\in\mathcal{V}$, where $k_{6}\in\mathbb{R}_{>0}$ is a
user-defined forgetting rate, $\Gamma_{i}\in\mathbb{R}^{p\times p}$
is a symmetric user-defined learning rate with strictly positive eigenvalues,
and proj$(\cdot)$ denotes a smooth projection operator as defined
in \cite[Appendix E]{Krstic1995}, which ensures $\hat{\theta}(t)\in\mathcal{B}_{\overline{\theta}}\triangleq\left\{ \theta\in\mathbb{R}^{p}:\left\Vert \theta\right\Vert \leq\overline{\theta}\right\} $
for all $t\in\mathbb{R}_{\geq0}$.

\subsection{Ensemble Analysis}

To aid in the stability analysis, the interaction matrix $\mathcal{H}\in\mathbb{R}^{3N\times3N}$
is defined as
\begin{equation}
\mathcal{H}\triangleq\left(L\otimes I_{3}\right)+\mathcal{C},\label{eq:Augmented Interaction Matrix}
\end{equation}
where $\mathcal{C}\triangleq\text{blkdiag}\left(b_{1}C_{1}^{\top}C_{1},\ldots,b_{N}C_{N}^{\top}C_{N}\right)\in\mathbb{R}^{3N\times3N}$.
Using (\ref{eq:Augmented Interaction Matrix}), (\ref{eq:relative position sychronization})
is expressed in an ensemble form as
\begin{equation}
\eta=\mathcal{H}e,\label{eq:ensembleEta}
\end{equation}
where $\eta\triangleq\left(\eta_{i}\right)_{i\in\mathcal{V}}\in\mathbb{R}^{3N}$
and $e\triangleq\left(e_{i}\right)_{i\in\mathcal{V}}\in\mathbb{R}^{3N}$.
Using (\ref{eq:Filtered Tracking Error}) and (\ref{eq:ensembleEta})
yields the useful expression
\begin{equation}
\mathcal{H}r=\dot{\eta}+k_{1}\eta,\label{eq:Hr}
\end{equation}
where $r\triangleq\left(r_{i}\right)_{i\in\mathcal{V}}\in\mathbb{R}^{3N}$.
Using (\ref{eq:velocity estimation error}), (\ref{eq:Filtered Estimation Error}),
(\ref{eq: Distributed Observer}), (\ref{eq:ensembleEta}) and (\ref{eq:Hr}),
(\ref{eq: controller-1}) is expressed in an ensemble form as
\begin{align}
u & =\mathbf{g}^{-1}\left(\Phi+k_{2}\mathcal{H}r-k_{2}\tilde{r}\right),\label{eq:ensemble controller}
\end{align}
where $u\triangleq\left(u_{i}\right)_{i\in\mathcal{V}}\in\mathbb{R}^{3N}$,
$\mathbf{g}^{-1}\triangleq I_{N}\otimes g^{-1}\in\mathbb{R}^{3N\times3N}$,
$\Phi\triangleq\left(\Phi_{i}\right)_{i\in\mathcal{V}}\in\mathbb{R}^{3N}$,
$\tilde{r}\triangleq\left(\tilde{r}_{i}\right)_{i\in\mathcal{V}}\in\mathbb{R}^{3N}$,
and $\rho\triangleq\left(\rho_{i}\right)_{i\in\mathcal{V}}\in\mathbb{R}^{3N}$.
Using (\ref{eq:Augmented Interaction Matrix}), (\ref{eq: Distributed Observer})
is expressed in an ensemble form as
\begin{align}
\dot{\hat{\eta}} & =\hat{\zeta},\nonumber \\
\dot{\hat{\zeta}} & =-\mathcal{H}\mathbf{g}u-(k_{3}^{2}-2)\tilde{\eta}-\left(2k_{3}+k_{4}+k_{5}\right)\rho,\label{eq:observer}
\end{align}
where $\hat{\eta}\triangleq\left(\hat{\eta}_{i}\right)_{i\in\mathcal{V}}\in\mathbb{R}^{3N}$,
$\hat{\zeta}\triangleq\left(\hat{\zeta}_{i}\right)_{i\in\mathcal{V}}\in\mathbb{R}^{3N}$,
$\mathbf{g}\triangleq I_{N}\otimes g\in\mathbb{R}^{3N\times3N}$,
and $\tilde{\eta}\triangleq\left(\tilde{\eta}_{i}\right)_{i\in\mathcal{V}}\in\mathbb{R}^{3N}$.
Furthermore, the time-derivative of (\ref{eq:rho implemented}) is
expressed in an ensemble form as
\begin{equation}
\dot{\rho}\triangleq\tilde{\eta}-(k_{3}+k_{4})\tilde{r}-k_{5}\rho.\label{eq:rho}
\end{equation}
Substituting (\ref{eq:ensemble controller}) into the ensemble representation
of (\ref{eq:Filtered Tracking Error Derivative-2}) yields
\begin{align}
\dot{r} & =\nabla_{\hat{\theta}}\Phi\widetilde{\theta}-k_{2}\mathcal{H}r+k_{2}\tilde{r}+k_{1}r-k_{1}^{2}e-\mathbf{w}+\tilde{f}+\Delta,\label{eq:rDot}
\end{align}
where $\nabla_{\hat{\theta}}\Phi\triangleq\text{blkdiag}\left(\nabla_{\hat{\theta}_{1}}\Phi_{1},\ldots,\nabla_{\hat{\theta}_{N}}\Phi_{N}\right)\in\mathbb{R}^{3N\times3N}$,
$\widetilde{\theta}\triangleq\left(\tilde{\theta}_{i}\right)_{i\in\mathcal{V}}\in\mathbb{R}^{3N}$,
$\mathbf{w}\triangleq\mathbf{1}_{N}\otimes\omega\in\mathbb{R}^{3N}$,
$\tilde{f}\triangleq\left(\tilde{f}_{i}\right)_{i\in\mathcal{V}}\in\mathbb{R}^{3N}$,
and $\Delta\triangleq\left(\Delta_{i}\right)_{i\in\mathcal{V}}\in\mathbb{R}^{3N}$.
Taking the time-derivative of (\ref{eq:Filtered Estimation Error}),
using (\ref{eq:tracking error acceleration}), (\ref{eq:LbDNN}),
(\ref{eq:TaylorApproximation}), (\ref{eq:Filtered Estimation Error}),
and (\ref{eq:ensembleEta}), and then substituting (\ref{eq:observer})
and (\ref{eq:rho}) into the ensemble representation of the resulting
expression yields
\begin{align}
\dot{\tilde{r}} & =\mathcal{H}\left(h-\mathbf{w}+\tilde{f}\right)+\left(k_{3}+k_{4}\right)\rho-\tilde{\eta}-k_{4}\tilde{r},\label{eq:rTildeDot}
\end{align}
where $h\triangleq\left(h_{i}(\kappa_{i})\right)_{i\in\mathcal{V}}\in\mathbb{R}^{3N}$.
The ensemble representation of (\ref{eq:adaptive update law}) is
expressed as
\begin{equation}
\dot{\hat{\theta}}=\left[\begin{array}{ccc}
\text{proj}\left(\Theta_{1},\hat{\theta}_{1},\overline{\theta}\right)^{\top} & \cdots & \text{proj}\left(\Theta_{N},\hat{\theta}_{N},\overline{\theta}\right)^{\top}\end{array}\right]^{\top},\label{eq:updateLaw}
\end{equation}
where $\hat{\theta}\triangleq\left(\hat{\theta}_{i}\right)_{i\in\mathcal{V}}\in\mathbb{R}^{pN}$
. Substituting (\ref{eq:updateLaw}) into the ensemble representation
of (\ref{eq: parameter estimation error}) yields
\begin{equation}
\dot{\widetilde{\theta}}=-\left[\begin{array}{ccc}
\text{proj}\left(\Theta_{1},\hat{\theta}_{1},\overline{\theta}\right)^{\top} & \cdots & \text{proj}\left(\Theta_{N},\hat{\theta}_{N},\overline{\theta}\right)^{\top}\end{array}\right]^{\top}.\label{eq:thetaTildeDot}
\end{equation}

This work focuses on spacecraft with limited sensing capabilities.
Accurate state estimation is necessary for spacecraft servicing missions,
particularly in scenarios where multiple spacecraft track and service
a target. Each spacecraft's sensing limitations may restrict it to
measuring only a subset of the target's states. Without information
sharing, the collective system may fail to reconstruct the target's
complete state. For instance, in a two-dimensional scenario, two spacecraft
that each measure only one degree of freedom and cannot communicate
would be limited to aligning collinearly with the target. However,
by sharing partial measurements, the spacecraft can collaborate to
achieve complete state reconstruction, enabling effective tracking
and servicing.

The trackability condition formalizes this idea, ensuring the spacecraft's
collective sensing provides sufficient information for accurate state
estimation. It guarantees stability in the closed-loop error system
by ensuring the eigenvalues of the matrix $\mathcal{H}$ are positive.
The following definition describes this concept.
\begin{defn}
(Trackability, \cite[Lemma 2]{Nino.Patil.ea2025}) A target agent
is said to be trackable if the following condition, known as the trackability
condition, is satisfied: $\text{rank}\left(\sum_{i\in\mathcal{V}}b_{i}C_{i}^{\top}C_{i}\right)=3$.
\end{defn}

\section{Stability Analysis}

Define the concatenated state vector $z:\mathbb{R}_{\geq0}\to\mathbb{R}^{\varphi}$
as $z\triangleq\left[\begin{array}{cccccc}
e^{\top} & r^{\top} & \tilde{\eta}^{\top} & \tilde{r}^{\top} & \rho^{\top} & \tilde{\theta}^{\top}\end{array}\right]^{\top}$, where $\varphi\triangleq(15+p)N$. Using (\ref{eq:Filtered Tracking Error}),
(\ref{eq:Filtered Estimation Error}), (\ref{eq:rho}), (\ref{eq:rDot}),
and (\ref{eq:thetaTildeDot}) yields
\begin{equation}
\dot{z}=\left[\begin{array}{c}
r-k_{1}e\\
\begin{array}{c}
\nabla_{\hat{\theta}}\Phi\widetilde{\theta}-k_{2}\mathcal{H}r+k_{2}\tilde{r}+k_{1}r-k_{1}^{2}e-\mathbf{w}+\tilde{f}+\Delta\end{array}\\
\tilde{r}-k_{3}\tilde{\eta}-\rho\\
\begin{array}{c}
\mathcal{H}\left(h-\mathbf{w}+\tilde{f}\right)+\left(k_{3}+k_{4}\right)\rho-\tilde{\eta}-k_{4}\tilde{r}\end{array}\\
\tilde{\eta}-\left(k_{3}+k_{4}\right)\tilde{r}-k_{5}\rho\\
-\left[\begin{array}{ccc}
\text{proj}\left(\Theta_{1},\hat{\theta}_{1},\overline{\theta}\right)^{\top} & \cdots & \text{proj}\left(\Theta_{N},\hat{\theta}_{N},\overline{\theta}\right)^{\top}\end{array}\right]^{\top}
\end{array}\right].\label{eq:CLES}
\end{equation}

By the Universal Approximation Theorem in \cite[Theorem 3.2]{Kidger2020},
the subsequent stability analysis requires ensuring $\kappa_{i}(t)\in\Omega$
for all $i\in\mathcal{V}$, for all $t\in\mathbb{R}_{\geq0}$. This
requirement is guaranteed to be satisfied by achieving a stability
result which constrains $z$ to a compact domain. Define the compact
domain bounding all system trajectories as
\begin{equation}
\mathcal{D}\triangleq\left\{ \iota\in\mathbb{R}^{\varphi}:\left\Vert \iota\right\Vert \leq\chi\right\} ,\label{eq:Local Domain}
\end{equation}
where $\chi\in\mathbb{R}_{>0}$ is a bounding constant.

By the continuous differentiability of $f_{k}$, it follows that $f_{k}$
is Lipschitz continuous over $\mathcal{D}$ for all $k\in\overline{\mathcal{V}}$.
Consequently, $h_{i}$ is Lipschitz continuous over $\mathcal{D}$
for all $i\in\mathcal{V}$. Hence, there exists a constant ${\tt D}\in\mathbb{R}_{\geq0}$
such that $\left\Vert h_{i}\right\Vert =\left\Vert f_{0}(\eta_{i},\hat{\zeta}_{i})-f_{i}(\eta_{i},\hat{\zeta}_{i})\right\Vert \leq{\tt D}$
for all $i\in\mathcal{V}$ and for all time. Furthermore, from (\ref{eq:LbDNN}),
there exists a constant $L_{\Phi}\in\mathbb{R}_{\geq0}$ such that
$\left\Vert \nabla_{\hat{\theta}}\Phi\right\Vert \leq L_{\Phi}$.
Similarly, by Assumption \ref{Asp: DNN Bounds} and the projection
operator, it holds that $\left\Vert \tilde{\theta}_{i}\right\Vert \leq\left\Vert \theta_{i}^{*}\right\Vert +\left\Vert \hat{\theta}_{i}\right\Vert \leq2\overline{\theta}$.
Therefore, there exists $\overline{\Delta}\triangleq2M\overline{\theta}^{2}+\overline{\varepsilon}\in\mathbb{R}_{>0}$
such that $\left\Vert \Delta\right\Vert =\left\Vert \left(\Delta_{i}\right)_{i\in\mathcal{V}}\right\Vert =\left(NR_{1,i}(\widetilde{\theta}_{i},\kappa_{i})+\varepsilon_{i}(\kappa_{i})\right)_{i\in\mathcal{V}}\leq N\left(\frac{M_{i}}{2}\left\Vert \tilde{\theta}_{i}\right\Vert ^{2}\right)_{i\in\mathcal{V}}+\overline{\varepsilon}\leq2NM\overline{\theta}^{2}+\overline{\varepsilon}=\overline{\Delta}$,
where $M\triangleq\max_{i\in\mathcal{V}}\left\{ M_{i}\right\} $,
for all $z\in\mathcal{D}$.

To facilitate the stability analysis, consider the Lyapunov function
candidate $V:\mathcal{D}\to\mathbb{R}_{\geq0}$ defined as
\begin{equation}
V(z)\triangleq\frac{1}{2}z^{\top}Pz,\label{eq:Lyapunov}
\end{equation}
where $P\triangleq\text{blkdiag}\left(I_{15N},\Gamma^{-1}\right)\in\mathbb{R}^{\varphi\times\varphi}$
and $\Gamma\triangleq\text{blkdiag}\left(\Gamma_{1},\ldots,\Gamma_{N}\right)\in\mathbb{R}^{pN\times pN}$.
By the Rayleigh quotient theorem (see \cite[Theorem 4.2.2]{Horn.Johnson1993}),
(\ref{eq:Lyapunov}) satisfies
\begin{equation}
\lambda_{1}\left\Vert z\right\Vert ^{2}\leq V(z)\leq\lambda_{2}\left\Vert z\right\Vert ^{2},\label{eq: RQ}
\end{equation}
where $\lambda_{1}\triangleq\frac{1}{2}\min\left\{ 1,\lambda_{\min}\left(\Gamma^{-1}\right)\right\} $
and $\lambda_{2}\triangleq\frac{1}{2}\max\left\{ 1,\lambda_{\max}\left(\Gamma^{-1}\right)\right\} $.
Based on the subsequent set definitions, let $\delta\triangleq\frac{\left(\overline{\Delta}+N\overline{\omega}+2LN^{2}\left(\overline{q}_{0}+\overline{\dot{q}}_{0}\right)\right)^{2}}{\underline{\lambda}_{\mathcal{H}}k_{2}}+\frac{\left(\overline{\lambda}_{\mathcal{H}}{\tt D}+2LN^{2}\overline{\lambda}_{\mathcal{H}}\left(\overline{q}_{0}+\overline{\dot{q}}_{0}\right)+N\overline{\omega}\overline{\lambda}_{\mathcal{H}}\right)^{2}}{2k_{4}}+k_{6}\overline{\theta}^{2}\overline{\lambda}_{\mathcal{J}}$
and $\lambda_{3}\triangleq\frac{1}{2}\min\left\{ \alpha_{i}\right\} _{i=\left[6\right]}$,
where
\begin{align*}
\alpha_{1} & =2-2\overline{\lambda}_{\mathcal{H}}L_{\Phi}-4LN\overline{\lambda}_{\mathcal{H}}\left(\overline{\lambda}_{\mathcal{H}}+1\right),\\
\alpha_{2} & =k_{2}\underline{\lambda}_{\mathcal{H}}-k_{2}-LN\left(k_{3}+2\left(2\overline{\lambda}_{\mathcal{H}}+1\right)+\left(1+k_{1}\right)\left(2\overline{\lambda}_{\mathcal{H}}+1\right)+2\right)-L_{\Phi}\left(1+\overline{\lambda}_{\mathcal{H}}\right)-2,\\
\alpha_{3} & =2-L_{\Phi}-LN\left(\overline{\lambda}_{\mathcal{H}}+1\right),\\
\alpha_{4} & =k_{4}-k_{2}-LN\left(\left(k_{3}+2+\left(1+k_{1}\right)\left(2\overline{\lambda}_{\mathcal{H}}+1\right)\right)+1\right)\overline{\lambda}_{\mathcal{H}}-L_{\Phi},\\
\alpha_{5} & =2k_{5}-LN\left(\overline{\lambda}_{\mathcal{H}}+1\right)-L_{\Phi},\\
\alpha_{6} & =k_{6}\underline{\lambda}_{\mathcal{J}}-L_{\Phi}\left(k_{3}+3\overline{\lambda}_{\mathcal{H}}+3\right),
\end{align*}
where $\overline{\lambda}_{(\cdot)}\triangleq\lambda_{\max}(\cdot)$
and $\underline{\lambda}_{(\cdot)}\triangleq\lambda_{\min}(\cdot)$.

For the dynamical system described by (\ref{eq:CLES}), the set of
stabilizing initial conditions $\mathcal{S}\subset\mathbb{R}^{\varphi}$
is defined as
\begin{equation}
\mathcal{S}\triangleq\left\{ \iota\in\mathbb{R}^{\varphi}:\left\Vert \iota\right\Vert \leq\sqrt{\frac{\lambda_{1}}{\lambda_{2}}\chi^{2}-\frac{\delta}{\lambda_{3}}}\right\} ,\label{eq:initial conditions}
\end{equation}
and the uniformly ultimately bounded (UUB) set $\mathcal{U}\subset\mathbb{R}^{\varphi}$
is defined as
\begin{equation}
\mathcal{U}\triangleq\left\{ \iota\in\mathbb{R}^{\varphi}:\left\Vert \iota\right\Vert \leq\sqrt{\frac{\lambda_{2}}{\lambda_{1}}\frac{\delta}{\lambda_{3}}}\right\} .\label{eq:equilibrium set}
\end{equation}
Furthermore, based on the subsequent stability analysis, the parameter
interaction matrix $\mathcal{J}\in\mathbb{R}^{pN\times pN}$ is defined
as
\begin{equation}
\mathcal{J}\triangleq(\mathcal{L}_{G}+I_{N})\otimes I_{p}.\label{eq:weightinteractionmatrix}
\end{equation}

\begin{lem}
\label{lem:trackability}If the target is trackable, then ${\tt z}^{\top}\mathcal{H}{\tt z}>0$
for any ${\tt z}\neq\mathbf{0}_{nN}$.
\begin{proof}
See \cite[Lemma 1]{Nino.Patil.ea2025}
\end{proof}
\end{lem}
\begin{thm}
Consider the dynamical system described by (\ref{eq:Dynamics}) and
(\ref{eq:Target Dynamics}). For any initial conditions of the states
$\left\Vert z(t_{0})\right\Vert \in\mathcal{S}$, the observer given
by (\ref{eq: Distributed Observer}), the controller given by (\ref{eq: controller-1}),
and the adaptation law given by (\ref{eq:adaptive update law}) ensure
that $z$ exponentially converges to $\mathcal{U}$ in the sense that
\begin{equation}
\left\Vert z(t)\right\Vert \leq\sqrt{\frac{\lambda_{2}}{\lambda_{1}}}\sqrt{\left\Vert z(t_{0})\right\Vert ^{2}{\rm e}^{-\frac{\lambda_{3}}{\lambda_{2}}(t-t_{0})}+\frac{\delta}{\lambda_{3}}\left(1-{\rm e}^{-\frac{\lambda_{3}}{\lambda_{2}}(t-t_{0})}\right)},\label{eq:solution-1}
\end{equation}
for all $t\in\left[t_{0},\infty\right)$, provided that $\lambda_{3}>0$,
$\chi>\sqrt{\frac{\lambda_{2}}{\lambda_{1}}\frac{\delta}{\lambda_{3}}}\sqrt{\frac{\lambda_{2}}{\lambda_{1}}+1}$,
Assumptions \ref{targetbounds}-\ref{Asp: DNN Bounds} hold, and the
target is trackable.
\end{thm}
\begin{proof}
Substituting (\ref{eq:CLES}) into the time-derivative of (\ref{eq:Lyapunov}),
using (\ref{eq: parameter estimation error}), and simplifying yields
\begin{align}
\dot{V}(z) & =-k_{1}e^{\top}e-k_{2}r^{\top}\mathcal{H}r-k_{3}\tilde{\eta}^{\top}\tilde{\eta}-k_{4}\tilde{r}^{\top}\tilde{r}-k_{5}\rho^{\top}\rho+\left(1-k_{1}^{2}\right)r^{\top}e+k_{2}r^{\top}\tilde{r}+k_{1}r^{\top}r\nonumber \\
 & +\widetilde{\theta}^{\top}\nabla_{\hat{\theta}}^{\top}\Phi r+\tilde{r}^{\top}\mathcal{H}\left(h-\mathbf{w}+\tilde{f}\right)+r^{\top}\left(\Delta-\mathbf{w}+\tilde{f}\right)\nonumber \\
 & -\tilde{\theta}^{\top}\Gamma^{-1}\left[\begin{array}{ccc}
\text{proj}\left(\Theta_{1},\overline{\theta}\right)^{\top} & \cdots & \text{proj}\left(\Theta_{N},\overline{\theta}\right)^{\top}\end{array}\right].\label{eq:Vdot}
\end{align}
Invoking \cite[Lemma E.1.IV]{Krstic1995}, using (\ref{eq:Augmented Interaction Matrix})
and the definition of $\mathcal{J}$ in (\ref{eq:weightinteractionmatrix}),
and using the ensemble representation of (\ref{eq: parameter estimation error})
yields
\begin{align}
-\tilde{\theta}^{\top}\Gamma^{-1}\begin{bmatrix}\text{proj}\left(\Theta_{1},\hat{\theta}_{1},\overline{\theta}\right)\\
\vdots\\
\text{proj}\left(\Theta_{N},\hat{\theta}_{N},\overline{\theta}\right)
\end{bmatrix} & \leq-\tilde{\theta}^{\top}\nabla_{\hat{\theta}}^{\top}\Phi\left(\hat{\zeta}+k_{1}\eta\right)+k_{6}\tilde{\theta}^{\top}\mathcal{\mathcal{J}}\theta^{\ast}-k_{6}\tilde{\theta}^{\top}\mathcal{\mathcal{J}}\tilde{\theta}.\label{eq:projection inequality}
\end{align}
Using (\ref{eq:projection inequality}), applying the triangle inequality
and the Cauchy-Schwarz inequality to the right-hand-side of (\ref{eq:Vdot})
and using the definitions of $\underline{\lambda}_{\mathcal{H}}$,
$\underline{\lambda}_{\mathcal{J}}$, $\overline{\lambda}_{\mathcal{J}}$,
$L_{\Phi}$, $\overline{\Delta}$, and $\overline{\omega}$ and Assumption
\ref{Asp: DNN Bounds} to the resulting expression yields
\begin{align}
\dot{V}(z) & \leq-\left\Vert e\right\Vert ^{2}-k_{2}\underline{\lambda}_{\mathcal{H}}\left\Vert r\right\Vert ^{2}-k_{3}\left\Vert \tilde{\eta}\right\Vert ^{2}-k_{4}\left\Vert \tilde{r}\right\Vert ^{2}-k_{5}\left\Vert \rho\right\Vert ^{2}-k_{6}\underline{\lambda}_{\mathcal{J}}\left\Vert \tilde{\theta}\right\Vert ^{2}+k_{6}\overline{\theta}\overline{\lambda}_{\mathcal{J}}\left\Vert \tilde{\theta}\right\Vert \nonumber \\
 & +k_{2}\left\Vert r\right\Vert \left\Vert \tilde{r}\right\Vert +\left\Vert r\right\Vert ^{2}+L_{\Phi}\left\Vert \widetilde{\theta}\right\Vert \left(\left\Vert r\right\Vert +\left\Vert \hat{\zeta}\right\Vert +\left\Vert \eta\right\Vert \right)\nonumber \\
 & +\overline{\lambda}_{\mathcal{H}}\left\Vert \tilde{r}\right\Vert \left(\left\Vert h\right\Vert +\left\Vert \tilde{f}\right\Vert +N\overline{\omega}\right)+\left\Vert r\right\Vert \left(\overline{\Delta}+N\overline{\omega}+\left\Vert \tilde{f}\right\Vert \right),\label{eq:Vdot-1}
\end{align}
for all $z\in\mathcal{D}$. Using the ensemble representation of (\ref{eq:Filtered Tracking Error}),
(\ref{eq:position estimation error}), (\ref{eq:ensembleEta}), and
(\ref{eq:Hr}) yields
\begin{equation}
\left\Vert \hat{\zeta}\right\Vert \leq\overline{\lambda}_{\mathcal{H}}\left\Vert e\right\Vert +\overline{\lambda}_{\mathcal{H}}\left\Vert r\right\Vert +k_{3}\left\Vert \tilde{\eta}\right\Vert +\left\Vert \tilde{r}\right\Vert +\left\Vert \rho\right\Vert .\label{eq:zetaHat}
\end{equation}
Using (\ref{eq:zetaHat}), the ensemble representation of (\ref{eq:tracking error}),
the triangle inequality, Assumption \ref{targetbounds}, and the definition
of $\overline{\lambda}_{\mathcal{H}}$ yields the bound
\begin{align}
\left\Vert \tilde{f}\right\Vert  & \leq LN\left(\left(4\overline{\lambda}_{\mathcal{H}}+2\right)\left\Vert e\right\Vert +\left(2\overline{\lambda}_{\mathcal{H}}+1\right)\left\Vert r\right\Vert +k_{3}\left\Vert \tilde{\eta}\right\Vert +\left\Vert \tilde{r}\right\Vert +\left\Vert \rho\right\Vert +2N\left(\overline{q}_{0}+\overline{\dot{q}}_{0}\right)\right).\label{eq:fTilde}
\end{align}
Using (\ref{eq:ensembleEta}) and (\ref{eq:zetaHat}) yields
\begin{align}
L_{\Phi}\left\Vert \widetilde{\theta}\right\Vert \left(\left\Vert r\right\Vert +\left\Vert \hat{\zeta}\right\Vert +k_{1}\left\Vert \eta\right\Vert \right) & \leq2\overline{\lambda}_{\mathcal{H}}L_{\Phi}\left\Vert \widetilde{\theta}\right\Vert \left\Vert e\right\Vert +L_{\Phi}\left(1+\overline{\lambda}_{\mathcal{H}}\right)\left\Vert \widetilde{\theta}\right\Vert \left\Vert r\right\Vert +k_{3}L_{\Phi}\left\Vert \widetilde{\theta}\right\Vert \left\Vert \tilde{\eta}\right\Vert \nonumber \\
 & +L_{\Phi}\left\Vert \widetilde{\theta}\right\Vert \left\Vert \tilde{r}\right\Vert +L_{\Phi}\left\Vert \widetilde{\theta}\right\Vert \left\Vert \rho\right\Vert .\label{eq:bound-1}
\end{align}
By using (\ref{eq:fTilde}) and (\ref{eq:bound-1}), (\ref{eq:Vdot-1})
is upper bounded as
\begin{align}
\dot{V}(z) & \leq-\left\Vert e\right\Vert ^{2}-k_{2}\underline{\lambda}_{\mathcal{H}}\left\Vert r\right\Vert ^{2}+\left\Vert r\right\Vert ^{2}+LN\left(2\overline{\lambda}_{\mathcal{H}}+1\right)\left\Vert r\right\Vert ^{2}-k_{3}\left\Vert \tilde{\eta}\right\Vert ^{2}-k_{4}\left\Vert \tilde{r}\right\Vert ^{2}\nonumber \\
 & -k_{5}\left\Vert \rho\right\Vert ^{2}-k_{6}\underline{\lambda}_{\mathcal{J}}\left\Vert \tilde{\theta}\right\Vert ^{2}+k_{2}\left\Vert r\right\Vert \left\Vert \tilde{r}\right\Vert +2\overline{\lambda}_{\mathcal{H}}L_{\Phi}\left\Vert \widetilde{\theta}\right\Vert \left\Vert e\right\Vert +L_{\Phi}\left(1+\overline{\lambda}_{\mathcal{H}}\right)\left\Vert \widetilde{\theta}\right\Vert \left\Vert r\right\Vert \nonumber \\
 & +L_{\Phi}\left\Vert \widetilde{\theta}\right\Vert \left\Vert \tilde{r}\right\Vert +L_{\Phi}\left\Vert \widetilde{\theta}\right\Vert \left\Vert \rho\right\Vert +LN\overline{\lambda}_{\mathcal{H}}\left(4\overline{\lambda}_{\mathcal{H}}+2\right)\left\Vert e\right\Vert \left\Vert \tilde{r}\right\Vert +k_{3}LN\overline{\lambda}_{\mathcal{H}}\left\Vert \tilde{\eta}\right\Vert \left\Vert \tilde{r}\right\Vert \nonumber \\
 & +LN\overline{\lambda}_{\mathcal{H}}\left\Vert \rho\right\Vert \left\Vert \tilde{r}\right\Vert +LN\left(4\overline{\lambda}_{\mathcal{H}}+2\right)\left\Vert e\right\Vert \left\Vert r\right\Vert +LNk_{3}\left\Vert \tilde{\eta}\right\Vert \left\Vert r\right\Vert +LN\left\Vert \tilde{r}\right\Vert \left\Vert r\right\Vert \nonumber \\
 & +LN\left\Vert \rho\right\Vert \left\Vert r\right\Vert +\overline{\lambda}_{\mathcal{H}}{\tt D}\left\Vert \tilde{r}\right\Vert +2LN^{2}\overline{\lambda}_{\mathcal{H}}\left(\overline{q}_{0}+\overline{\dot{q}}_{0}\right)\left\Vert \tilde{r}\right\Vert +N\overline{\omega}\overline{\lambda}_{\mathcal{H}}\left\Vert \tilde{r}\right\Vert +k_{3}L_{\Phi}\left\Vert \widetilde{\theta}\right\Vert \left\Vert \tilde{\eta}\right\Vert \nonumber \\
 & +\overline{\Delta}\left\Vert r\right\Vert +N\overline{\omega}\left\Vert r\right\Vert +2LN^{2}\left(\overline{q}_{0}+\overline{\dot{q}}_{0}\right)\left\Vert r\right\Vert +k_{6}\overline{\theta}\overline{\lambda}_{\mathcal{J}}\left\Vert \tilde{\theta}\right\Vert +LN\overline{\lambda}_{\mathcal{H}}\left\Vert \tilde{r}\right\Vert ^{2},\label{eq:vDot}
\end{align}
for all $z\in\mathcal{D}$. By using Young's inequality and completing
the square, (\ref{eq:vDot}) is upper bounded as{\small{}
\begin{align}
\dot{V}(z) & \leq-\left(1-\overline{\lambda}_{\mathcal{H}}L_{\Phi}-2LN\overline{\lambda}_{\mathcal{H}}\left(\overline{\lambda}_{\mathcal{H}}+1\right)\right)\left\Vert e\right\Vert ^{2}\nonumber \\
 & -\frac{1}{2}\left(k_{2}\underline{\lambda}_{\mathcal{H}}-k_{2}-LN\left(k_{3}+2\left(2\overline{\lambda}_{\mathcal{H}}+1\right)+\left(1+k_{1}\right)\left(2\overline{\lambda}_{\mathcal{H}}+1\right)+2\right)\right.\nonumber \\
 & \left.-L_{\Phi}\left(1+\overline{\lambda}_{\mathcal{H}}\right)-2\right)\left\Vert r\right\Vert ^{2}\nonumber \\
 & -\frac{k_{3}}{2}\left(2-L_{\Phi}-LN\left(\overline{\lambda}_{\mathcal{H}}+1\right)\right)\left\Vert \tilde{\eta}\right\Vert ^{2}\nonumber \\
 & -\frac{1}{2}\left(k_{4}-k_{2}-LN\left(\left(k_{3}+2+\left(1+k_{1}\right)\left(2\overline{\lambda}_{\mathcal{H}}+1\right)\right)+1\right)\overline{\lambda}_{\mathcal{H}}-L_{\Phi}\right)\left\Vert \tilde{r}\right\Vert ^{2}\nonumber \\
 & -\frac{1}{2}\left(2k_{5}-LN\left(\overline{\lambda}_{\mathcal{H}}+1\right)-L_{\Phi}\right)\left\Vert \rho\right\Vert ^{2}\nonumber \\
 & -\frac{1}{2}\left(k_{6}\underline{\lambda}_{\mathcal{J}}-L_{\Phi}\left(k_{3}+3\overline{\lambda}_{\mathcal{H}}+3\right)\right)\left\Vert \widetilde{\theta}\right\Vert ^{2}+k_{6}\overline{\theta}^{2}\overline{\lambda}_{\mathcal{J}}\nonumber \\
 & +\frac{\left(\overline{\Delta}+N\overline{\omega}+2LN^{2}\left(\overline{q}_{0}+\overline{\dot{q}}_{0}\right)\right)^{2}}{\underline{\lambda}_{\mathcal{H}}k_{2}}+\frac{\left(\overline{\lambda}_{\mathcal{H}}{\tt D}+2LN^{2}\overline{\lambda}_{\mathcal{H}}\left(\overline{q}_{0}+\overline{\dot{q}}_{0}\right)+N\overline{\omega}\overline{\lambda}_{\mathcal{H}}\right)^{2}}{2k_{4}},\label{eq:vDot-1}
\end{align}
}for all $z\in\mathcal{D}$. Using the definitions of $\lambda_{3}$
and $z$ yields
\begin{equation}
\dot{V}(z)\leq-\lambda_{3}\left\Vert z\right\Vert ^{2}+\delta,\label{eq:vDot-2}
\end{equation}
for all $z\in\mathcal{D}$. From (\ref{eq: RQ}), it follows that
$-\lambda_{3}\left\Vert z\right\Vert ^{2}\leq-\frac{\lambda_{3}}{\lambda_{2}}V(z)$.
Therefore, (\ref{eq:vDot-2}) is upper bounded as
\begin{equation}
\dot{V}(z)\leq-\frac{\lambda_{3}}{\lambda_{2}}V(z)+\delta,\label{eq:vDot-3}
\end{equation}
for all $z\in\mathcal{D}$. Solving the differential inequality given
by (\ref{eq:vDot-3}) yields
\begin{equation}
V\left(z(t)\right)\leq V\left(z(t_{0})\right){\rm e}^{-\frac{\lambda_{3}}{\lambda_{2}}(t-t_{0})}+\frac{\lambda_{2}\delta}{\lambda_{3}}\left(1-{\rm e}^{-\frac{\lambda_{3}}{\lambda_{2}}(t-t_{0})}\right),\label{eq:vsol}
\end{equation}
for all $z\in\mathcal{D}$. From (\ref{eq: RQ}), it follows that
$V\left(z(t_{0})\right){\rm e}^{-\frac{\lambda_{3}}{\lambda_{2}}(t-t_{0})}\leq\lambda_{2}\left\Vert z(t_{0})\right\Vert ^{2}{\rm e}^{-\frac{\lambda_{3}}{\lambda_{2}}(t-t_{0})}$
and also $\lambda_{1}\left\Vert z(t)\right\Vert ^{2}\leq V\left(z(t)\right)$.
Therefore, (\ref{eq:vsol}) can be bounded as
\begin{equation}
\lambda_{1}\left\Vert z(t)\right\Vert ^{2}\leq\lambda_{2}\left\Vert z(t_{0})\right\Vert ^{2}{\rm e}^{-\frac{\lambda_{3}}{\lambda_{2}}(t-t_{0})}+\frac{\lambda_{2}\delta}{\lambda_{3}}\left(1-{\rm e}^{-\frac{\lambda_{3}}{\lambda_{2}}(t-t_{0})}\right),\label{eq:vsol-1}
\end{equation}
for all $z\in\mathcal{D}$. Solving (\ref{eq:vsol-1}) for $\left\Vert z(t)\right\Vert $
yields
\begin{equation}
\left\Vert z(t)\right\Vert \leq\sqrt{\frac{\lambda_{2}}{\lambda_{1}}}\sqrt{\left\Vert z(t_{0})\right\Vert ^{2}{\rm e}^{-\frac{\lambda_{3}}{\lambda_{2}}(t-t_{0})}+\frac{\delta}{\lambda_{3}}\left(1-{\rm e}^{-\frac{\lambda_{3}}{\lambda_{2}}(t-t_{0})}\right)},\label{eq:solution}
\end{equation}
for all $z\in\mathcal{D}$. The UUB set defined by (\ref{eq:equilibrium set})
is obtained by taking the limit as $t\to\infty$ of the right-hand
side of (\ref{eq:solution}), yielding $\lim_{t\to\infty}\left\Vert z(t)\right\Vert \leq\sqrt{\frac{\lambda_{2}}{\lambda_{1}}\frac{\delta}{\lambda_{3}}}$,
or $z(t)\in\mathcal{U}$ as $t\to\infty$. The set of stabilizing
initial conditions defined by (\ref{eq:initial conditions})\textbf{
}follows from an upper bound on (\ref{eq:solution}) given by $\left\Vert z(t)\right\Vert \leq\sqrt{\frac{\lambda_{2}}{\lambda_{1}}}\sqrt{\left\Vert z(t_{0})\right\Vert ^{2}+\frac{\delta}{\lambda_{3}}}$.
From the definition of $\mathcal{D}$, the condition $z\in\mathcal{D}$
holds if and only if $\sqrt{\frac{\lambda_{2}}{\lambda_{1}}}\sqrt{\left\Vert z(t_{0})\right\Vert ^{2}+\frac{\delta}{\lambda_{3}}}\leq\chi$,
which is equivalent to $\left\Vert z(t_{0})\right\Vert \leq\sqrt{\frac{\lambda_{1}}{\lambda_{2}}\chi^{2}-\frac{\delta}{\lambda_{3}}}$,
or $\left\Vert z(t_{0})\right\Vert \in\mathcal{S}$. For $\mathcal{S}$
to be nonempty and $\mathcal{U}\subset\mathcal{S}$, the feasibility
condition $\chi>\sqrt{\frac{\lambda_{2}}{\lambda_{1}}\frac{\delta}{\lambda_{3}}}\sqrt{\frac{\lambda_{2}}{\lambda_{1}}+1}$
must hold. Let $\Upsilon$, in (\ref{eq:Omega}), be defined as $\Upsilon\triangleq\left(\overline{\lambda}_{\mathcal{H}}\left(k_{1}+1\right)+k_{3}+2\right)\chi$.
By the definition of $\Omega$ in (\ref{eq:Omega}), it follows that
$\iota\in\mathcal{D}$ implies $\left\Vert \iota\right\Vert \leq\chi\leq\Upsilon$
for all time. To establish that $\kappa_{i}\in\Omega$ for all $i\in\mathcal{V}$,
for ensuring the universal approximation property holds, observe that
$\left\Vert \kappa_{i}\right\Vert \leq\sqrt{\left\Vert \eta_{i}\right\Vert ^{2}+\left\Vert \hat{\zeta}_{i}\right\Vert ^{2}}\leq\left\Vert \eta\right\Vert +\left\Vert \hat{\zeta}\right\Vert \leq\left(\overline{\lambda}_{\mathcal{H}}\left(k_{1}+1\right)+k_{3}+2\right)\chi=\Upsilon$.
This confirms $\kappa_{i}\in\Omega$ for all $i\in\mathcal{V}$ and
for all time. Consequently, $\mathcal{U}\subset\mathcal{S}\subset\mathcal{D}$,
ensuring that the universal approximation property and Lipschitz property
are satisfied, the state remains bounded, and the trajectories converge
to a nonempty domain that is strictly smaller than the set of stabilizing
initial conditions.

Since $\left\Vert z\right\Vert \leq\chi$ implies $\left\Vert e\right\Vert ,\left\Vert r\right\Vert ,\left\Vert \tilde{\eta}\right\Vert ,\left\Vert \tilde{r}\right\Vert ,\left\Vert \rho\right\Vert ,\left\Vert \tilde{\theta}\right\Vert \leq\chi$,
the states $e$, $r$, $\tilde{\eta}$, $\tilde{r}$, $\rho$ and
$\tilde{\theta}$ remain bounded. As $\kappa_{i}\in\Omega$, $\kappa_{i}$
is also bounded for all $i\in\mathcal{V}$. The boundedness of $\hat{\theta}_{i}$,
enforced by the projection operator, ensures $\Phi_{i}(\kappa_{i},\hat{\theta}_{i})$
is bounded for all $i\in\mathcal{V}$. The boundedness of $e_{i}$
implies $\eta_{i}$ is bounded, as given by (\ref{eq:relative position sychronization}),
for all $i\in\mathcal{V}$. Similarly, boundedness of $e$, $r$,
$\tilde{\eta}$, $\tilde{r}$, and $\rho$ ensures $\hat{\zeta}$
is bounded by (\ref{eq:zetaHat}). Since $e$ is bounded and $q_{0}$
is bounded by Assumption \ref{targetbounds}, $q_{i}$ is bounded
for all $i\in\mathcal{V}$. Likewise, boundedness of $r$, and $\dot{q}_{0}$
by Assumption \ref{targetbounds}, implies $\dot{q}_{i}$ is bounded
for all $i\in\mathcal{V}$. Consequently, the boundedness of $g_{i}^{-1}$,
$\Phi_{i}$, $\eta_{i}$, $\hat{\zeta}_{i}$, $\tilde{\eta}$, and
$\rho$ ensures $u_{i}$ is bounded, as given by (\ref{eq: controller-1}),
for all $i\in\mathcal{V}$. Since $g_{i}$, $u_{i}$, $\tilde{\eta}_{i}$,
$\rho_{i}$, and $\hat{\zeta}_{i}$ are bounded, the observer states
$\dot{\hat{\eta}}_{i}$ and $\dot{\hat{\zeta}}_{i}$ are bounded by
(\ref{eq: Distributed Observer}), for all $i\in\mathcal{V}$. Similarly,
boundedness of $\Phi_{i}$, $\eta_{i}$, $\hat{\zeta}_{i}$, and $\hat{\theta}_{i}$
ensures $\dot{\hat{\theta}}_{i}$ is bounded for all $i\in\mathcal{V}$.
Therefore, since $u_{i}$, $\rho_{i}$, $\dot{\hat{\eta}}_{i}$, $\dot{\hat{\zeta}}$,
and $\dot{\hat{\theta}}_{i}$ are bounded for all $i\in\mathcal{V}$,
all implemented signals remain bounded for all time.
\end{proof}

\section{Simulation}

\begin{figure}
\begin{centering}
\includegraphics[width=0.35\columnwidth]{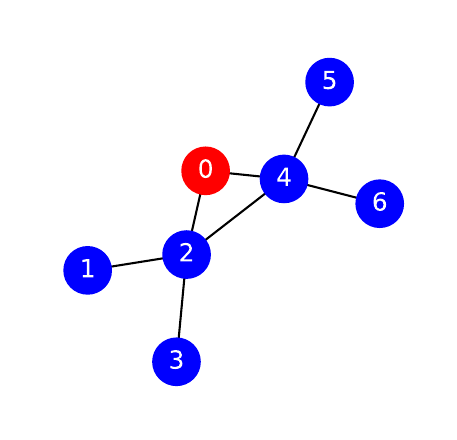}
\par\end{centering}
\caption{\label{fig:Plots-2}Communication Topology}
\end{figure}
In addition to the stability analysis, a simulation is included to
provide empirical evidence of the controller performance. The scenario
involves a network of $N=6$ servicer spacecraft tasked with tracking
a single defunct spacecraft, simulating an approach and servicing
operation. The communication topology governing the multi-spacecraft
network is depicted in Fig. \ref{fig:Plots-2}. The simulation spans
a duration of 360 seconds.

\begin{table}
\centering{}\caption{\label{tab:Result-Comparison}Spacecraft Parameters}
{\tiny{}}%
\begin{tabular}{|c|c|c|}
\hline 
Spacecraft Index & Mass (kg) & Cross-Sectional Area $(m^{2})$\tabularnewline
\hline 
\hline 
\multirow{1}{*}{$q_{0}$} & \multirow{1}{*}{10,000} & 1,000\tabularnewline
\hline 
$q_{1}$ & 640.1074 & 10.1267\tabularnewline
\hline 
$q_{2}$ & 1,087.4786 & 17.9324\tabularnewline
\hline 
$q_{3}$ & 25.1687 & 20.4416\tabularnewline
\hline 
$q_{4}$ & 470.9405 & 27.4020\tabularnewline
\hline 
$q_{5}$ & 241.4649 & 21.5405\tabularnewline
\hline 
$q_{6}$ & 161.1994 & 34.5757\tabularnewline
\hline 
\end{tabular}
\end{table}
The defunct spacecraft is initialized in a near-Earth elliptical orbit,
characterized by a periapsis altitude of $300\times10^{3}$ meters,
an apoapsis altitude of $700\times10^{3}$ meters, and an inclination
angle $\frac{\pi}{6}$ radians. The semi-major axis, $a$, is calculated
as the arithmetic mean of the periapsis and apoapsis altitudes. The
initial orbital velocity is derived from the vis-viva equation, $v=\sqrt{\mu\left(\frac{2}{r_{\text{periapsis}}}-\frac{1}{a}\right)}$,
where $\mu$ denotes the standard gravitational parameter of Earth.

\begin{figure}
\begin{centering}
\includegraphics[width=1\columnwidth]{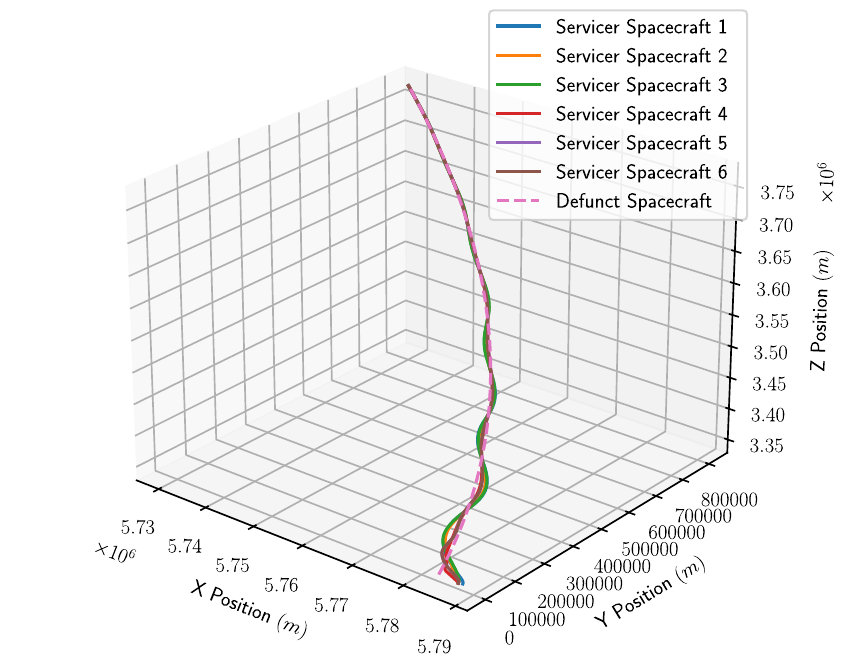}
\par\end{centering}
\caption{\label{fig:Plots}3D visualization of the trajectories of all servicer
spacecraft (solid lines) and the defunct spacecraft (dotted line).
For visual clarity, only the first 60 seconds of the simulation.}
\end{figure}
The initial positions of the servicer spacecraft are offset from the
defunct spacecraft's initial conditions. These offsets are generated
randomly, with radial distances uniformly distributed in the range
$U\left(2500,5000\right)$ meters, and angular deviations in azimuth
($\gamma$) and elevation ($\phi$) sampled from $U\left(-0.5,0.5\right)$
radians. Each servicer spacecraft's initial velocity matches the defunct
spacecraft's orbital velocity to ensure relative dynamics are dominated
by initial positional offsets.

\begin{figure}
\begin{centering}
\includegraphics[width=1\columnwidth]{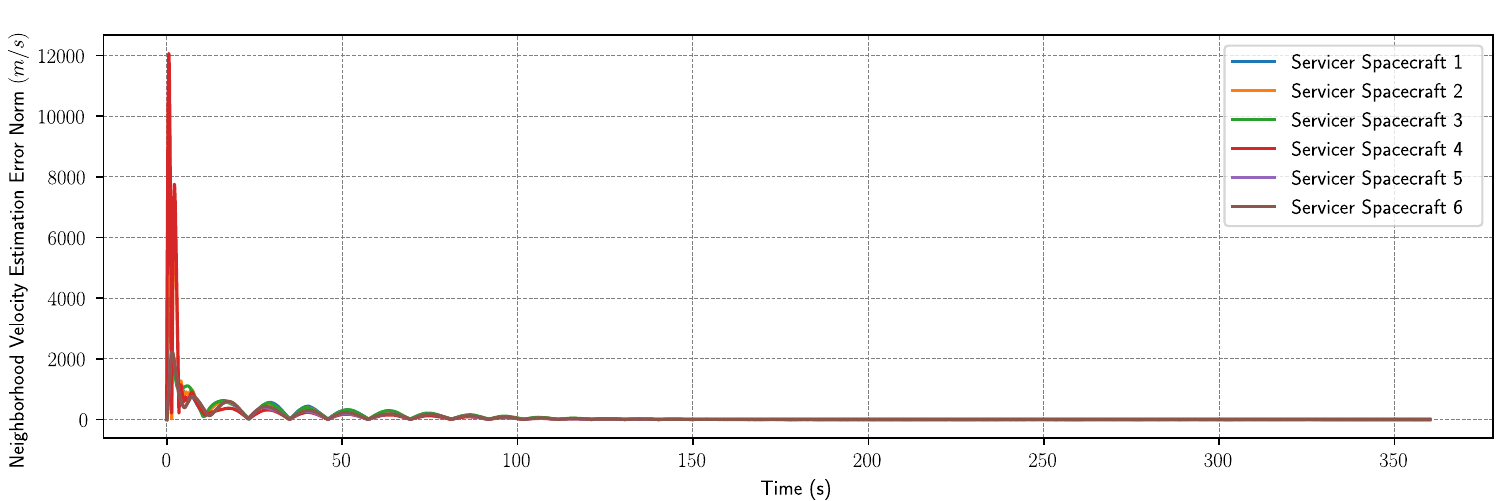}
\par\end{centering}
\caption{\label{fig:Plots-1}Plot showing the norm of the neighborhood velocity
estimation error $\left(\left\Vert \tilde{\zeta}_{i}\right\Vert \right)$
over time for all servicer spacecraft.}
\end{figure}
Each servicer spacecraft employs an output matrix $C_{i}\in\mathbb{R}^{p_{i}\times3}$,
where the entries are uniformly sampled from $U\left(-2,2\right)$.
The matrix row dimension, $p_{i}$, is randomly selected from $p_{i}\in\left\{ 1,2,3,4\right\} $.
The specific matrices selected for this simulation are as follows:
\begin{align*}
C_{1} & =\left[\begin{array}{ccc}
1.5124 & -1.8904 & 0.6818\end{array}\right],\ C_{2}=\left[\begin{array}{ccc}
-0.2772 & 1.7565 & 1.1135\\
0.8638 & 1.2110 & -1.6287
\end{array}\right],\\
C_{3} & =\left[\begin{array}{ccc}
1.5055 & 1.5784 & -1.6598\\
-1.8437 & -1.3206 & 1.5125
\end{array}\right],\ C_{4}=\left[\begin{array}{ccc}
0.3722 & 0.6866 & -0.3528\\
-1.2097 & -0.8414 & -1.4315\\
1.1332 & -0.3498 & -1.8633
\end{array}\right],\\
C_{5} & =\left[\begin{array}{ccc}
1.9554 & 0.9926 & -0.8782\\
1.1571 & -1.5870 & -0.2084\\
1.6343 & -0.8255 & -0.8488
\end{array}\right],\ C_{6}=\left[\begin{array}{ccc}
0.0991 & -1.6655 & 1.6674\end{array}\right].
\end{align*}

\begin{figure}
\begin{centering}
\includegraphics[width=1\columnwidth]{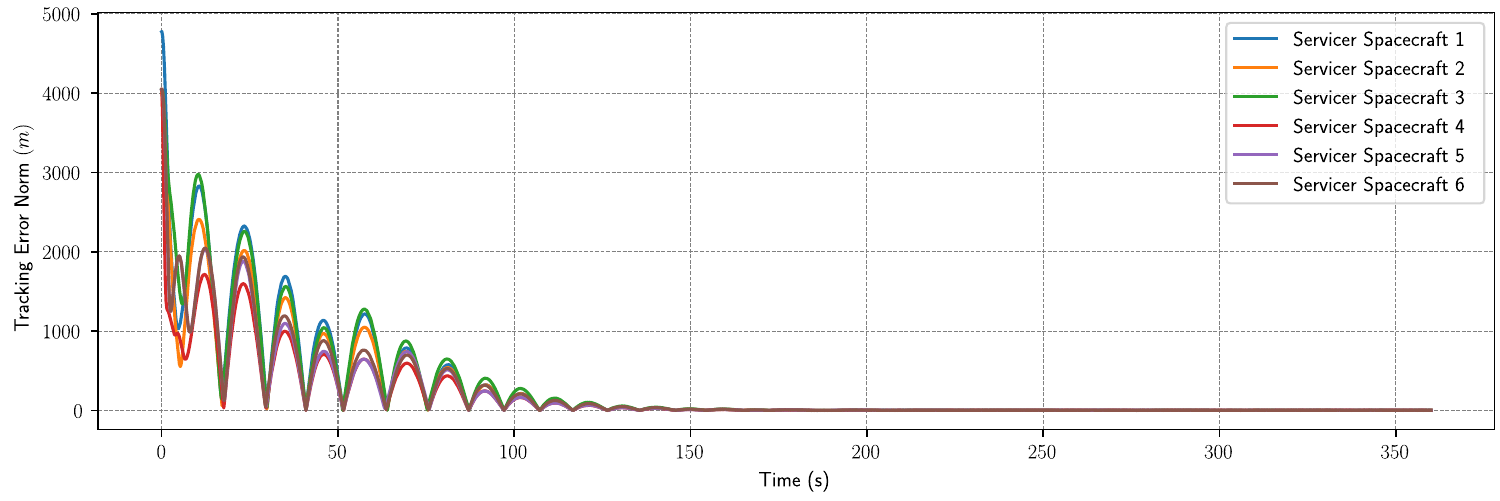}
\par\end{centering}
\caption{\label{fig:Plots-1-1}Plot showing the norm of the tracking error
$\left(\left\Vert e_{i}\right\Vert \right)$ over time for all servicer
spacecraft.}
\end{figure}
Each servicer spacecraft employs an Lb-DNN consisting of 4 hidden
layers with 4 neurons per hidden layer. This configuration results
in a total of 103 weights per spacecraft. The Lb-DNNs use the $\tanh$
activation function in the output layer, while the hidden layers employ
the Swish activation function \cite{Ramachandran.Zoph.ea2017}. The
weights are initialized using the Kaiming He initialization method
\cite{He.Zhang.ea2015} to leverage the Swish activation's smooth
approximation of the ReLU function. The control gains are selected
as $k_{i}=0.65$ for $i\in\{1,\ldots,5\}$, $k_{6}=0.0001$, and $\Gamma_{i}=0.01\cdot I_{103}$
for $i\in\{1,\ldots,6\}$.

The dynamics governing the servicer and defunct spacecraft follow
the nonlinear equations detailed in (\ref{eq:Dynamics}). For the
servicer spacecraft, masses and cross-sectional areas are sampled
from the uniform distributions $U\left(25,1500\right)\ kg$ \LyXThinSpace{}
and $U\left(1,50\right)\ m^{2}$, respectively. The properties of
all spacecraft are listed in Table \ref{tab:Result-Comparison}. The
measurement model follows (\ref{eq:measurement model}) and includes
additive Gaussian noise with zero mean and a standard deviation of
$\sqrt{0.5}$ meters. 

A 3D visualization of the trajectories of all servicer spacecraft
and the defunct spacecraft is shown in Fig. \ref{fig:Plots}. The
visualization is limited to the first 60 seconds of the simulation
to provide a detailed view of the initial trajectory dynamics. This
depiction illustrates the cooperative behavior of the servicer spacecraft
as they maneuver within the operational region.

The performance of the neighborhood velocity estimation is shown in
Fig. \ref{fig:Plots-1}, which shows the norm of the estimation error
$\left\Vert \tilde{\zeta}_{i}\right\Vert $ over time for all servicer
spacecraft. The results indicate that all servicer spacecraft achieve
steady-state estimation error values of approximately $5\ m/s$ after
about 150 seconds. The initial transient spikes observed within the
first 10 seconds are attributed to the zero-value initialization of
the observer. This transient response could be mitigated by initializing
the observer with prior estimates of the velocity, thereby improving
the convergence characteristics.

The tracking performance of the servicer spacecraft is shown in Fig.
\ref{fig:Plots-1-1}, which shows the norm of the tracking error $\left\Vert e_{i}\right\Vert $
over time. The plot shows that all servicer spacecraft achieve steady-state
tracking error values of approximately $5$ meters after approximately
200 seconds. This level of precision is considered sufficient to initiate
servicing operations. The oscillatory behavior observed during the
steady-state phase can be attributed to the choice of control gains.
While adjustments to the gain parameters could reduce oscillations,
such modifications would likely increase control effort, necessitating
a trade-off between control performance and energy efficiency. 

\section{Conclusion}

This work addresses challenges in spacecraft servicing by introducing
a distributed state estimation and tracking framework that relies
solely on relative position measurements and operates effectively
under partial state information. The proposed $\rho$-filter reconstructs
unknown states using locally available data, while a Lyapunov-based
deep neural network adaptive controller compensates for uncertainties
arising from unknown spacecraft dynamics. To ensure a well-posed collaborative
spacecraft regulation problem, a trackability condition is established.
Lyapunov-based stability analysis guarantees exponential error convergence
in state estimation and spacecraft regulation to a neighborhood of
the origin, contingent on this condition. Empirical validation is
demonstrated through a simulation involving a network of six servicer
spacecraft tasked with tracking a single defunct spacecraft. The servicer
spacecraft achieve steady-state tracking errors of approximately 5
meters within 200 seconds, meeting precision requirements for initiating
servicing operations.

\section*{Statements and Declarations}

\subsection*{Funding}
This research is supported in part by AFOSR grant FA9550-19-1-0169. Any opinions, findings, and conclusions or recommendations expressed in this material are those of the author(s) and do not necessarily reflect the views of the sponsoring agency.

\subsection*{Conflict of Interest}
On behalf of all authors, the corresponding author states that there
is no conflict of interest.

\subsection*{Consent for Publication}
This work has been approved for public release: distribution unlimited. Case Number AFRL-2024-5871.

\subsection*{Code Availability}
The code used in this study is available on request.

\bibliographystyle{sn-nature}
\bibliography{sources.bib}

\end{document}